\documentclass[graybox]{svmult}

\usepackage{amssymb,latexsym,amsfonts,amsmath}
\usepackage{graphics,graphicx,epic,shadow,color}
\usepackage{dsfont}
\newcommand{\Z}{\mathbb{Z}}
\newcommand{\R}{\mathbb{R}}
\newcommand{\N}{\mathbb{N}}
\newcommand{\C}{\mathcal{C}}

\newcommand{\A}{\mathbb{A}}

\newcommand{\pw}{ {\star} }
\newcommand{\supp}{\operatorname{supp}}

\newcommand{\sinc}{\operatorname{sinc}}

\def\abs#1{| #1 |}
\def\norm#1{\| #1 \|}
\def\set#1{\{ #1 \}}
\def\ip#1#2{\langle #1,#2 \rangle}

\usepackage{ paralist, xspace, subfig, chngcntr, float, extdash, varioref, array, soul, stmaryrd,  amsxtra, csquotes, cancel, xparse, units, etoolbox}
\usepackage{tikz}
\usetikzlibrary{
	matrix,%
	patterns, %
	backgrounds,%
	arrows,%
	chains,%
	positioning,%
	scopes,%
	decorations.pathreplacing,%
	decorations.shapes,%
	decorations.markings,%
	calc,%
	shapes.misc,%
	shapes
}
\usepackage{pgfplots}
\usepgfplotslibrary{patchplots}

\pgfplotsset{width=7cm,compat=newest}

\pgfdeclarelayer{myback}
\pgfsetlayers{myback,background,main}
\tikzset{
show curve controls/.style={
    decoration={
      show path construction,
      curveto code={
      \draw [blue, dashed]
          (\tikzinputsegmentfirst)    -- (\tikzinputsegmentsupporta)
          node [at end, draw, solid, red, inner sep=2pt]{};
        \draw [blue, dashed]
          (\tikzinputsegmentsupportb) -- (\tikzinputsegmentlast)
          node [at start, draw, solid, red, inner sep=2pt]{};
          \fill [black] (\tikzinputsegmentfirst) circle (2pt);
          \fill [red] (\tikzinputsegmentlast) circle (2pt);
      }
    },decorate
  }
} 

\colorlet{lightred}{red!60}
\colorlet{lightblue}{blue!30}

\newcommand{\mypatternnomarks}[2] 
{ 
\draw[style=help lines] (0,0) grid (#1,#1);
\foreach \m/\n in {#2}
	\filldraw[fill=mycolor, draw=black, shift={(\m-1,\n-1)}] (0,0) rectangle (1,1) node (0.5,0.5) {};  
\draw[->] (-0.2,0) -- (#1+0.3,0) node[right] {$(t,\nu)$};
\draw[->] (0,-0.2) -- (0,#1+0.3) node[above] {$(t',\nu') $};
}
\newcommand{\mypattern}[2] 
{	\pgfmathsetmacro\NN{#1 * #1};
	\draw[style=help lines] (0,0) grid (\NN,\NN);
	\foreach \m/\n in {#2}
		\filldraw[fill=mycolor, draw=black, shift={(\m-1,\n-1)}] (0,0) rectangle (1,1) node (0.5,0.5) {};  
	\draw[style=dashed, step=#1] (0.1,0.1) grid (\NN,\NN);
	\pgfmathparse{#1 - 1}
	\let\Lminusone\pgfmathresult	
	\foreach \m in {0,...,\Lminusone}
	\foreach \n in {0,...,\Lminusone}
	{	
		\pgfmathsetmacro\c{int(#1*\m+\n+1)}
		\node[below] at (\c-0.5,0) {${\eta}_{\m\n}$}; 
		\node[left] at (0,\c-0.5)  { ${\eta}_{\m\n}$};  
	}
	\draw[->] (-0.2,0) -- (\NN+0.3,0) node[right] {$(t,\nu)$};
	\draw[->] (0,-0.2) -- (0,\NN+0.3) node[above] {$(t',\nu') $};
}
\newcommand{\mysympatternwithdots}[6] 
{ \footnotesize
	\draw[style=help lines] (0,0) grid (#2-0.01,#2-0.01);
	\draw[style=dashed, step=#1] (0.1,0.1) grid (#2,#2);
	\foreach \m/\n in {#3}
		\filldraw[fill=mycolor, draw=black, shift={(\m-1,\n-1)}] (0,0) rectangle (1,1) node (0.5,0.5) {};  
	\foreach \m/\n in {#4}
	{	\filldraw[fill=mycolor, draw=black, shift={(\m-1,\n-1)}] (0,0) rectangle (1,1) node (0.5,0.5) {};  
		\filldraw[fill=mycolor, draw=black, shift={(\n-1,\m-1)}] (0,0) rectangle (1,1) node (0.5,0.5) {};  
	}
	\pgfmathparse{#1 - 1}
	\let\Lminusone\pgfmathresult	
	\foreach \m in {0,...,\Lminusone}
	\foreach \n in {0,...,\Lminusone}
	{	
		\pgfmathsetmacro\c{int(#1*\m+\n+1)}
		\pgfmathsetmacro\d{int(#2)<int(#1*\m+\n+1)}
		\ifthenelse{\not\equal{\d}{0.0}}{\breakforeach}{%
			\node[below] at (\c-0.5,0) (eta\m\n) {${\eta}_{\m\n}$}; 
			\node[left] at (0,\c-0.5) (leta\m\n) { ${\eta}_{\m\n}$};  
		}
	}
	\draw[->] (-0.2,0) -- (#2+0.3,0) node[right] {$(t,\nu)$};
	\draw[->] (0,-0.2) -- (0,#2+0.3) node[above] {$(t',\nu') $};
}
\newcommand{\mysympattern}[3] 
{ \footnotesize
	\mypattern{#1}{#2};
	\foreach \m/\n in {#3}
	{	\filldraw[fill=mycolor, draw=black, shift={(\m-1,\n-1)}] (0,0) rectangle (1,1) node (0.5,0.5) {};  
		\filldraw[fill=mycolor, draw=black, shift={(\n-1,\m-1)}] (0,0) rectangle (1,1) node (0.5,0.5) {};  
	}
}
\newcommand{\mynonsympattern}[3] 
{ \footnotesize
\mypattern{#1}{#2}
\foreach \m/\n in {#3}
	\filldraw[fill=mycolor, draw=black, shift={(\m-1,\n-1)}] (0,0) rectangle (1,1) node (0.5,0.5) {};  
}

\newcommand{\myaxes}[1]
{
\pgfmathparse{int(#1)}
\let\L\pgfmathresult
\draw[->] (-0.2,0) -- (\L+0.2,0) node[below] {$(t,\nu)$};
\draw[->] (0,-0.2) -- (0,\L+0.2) node[above] {$(t',\nu')$};
}

\newcommand{\myticks}[1]
{
\foreach \x in {#1}
{   \draw[shift={(\x,0)}] (0pt,2pt) -- (0pt,-2pt); 
    \draw[shift={(0,\x)}] (2pt,0pt) -- (-2pt,0pt); 
}
}

\newcommand{\tensorsupport}
{
\myaxes{8}
\myticks{2,4.5,6,7}
\begin{scope}[dashed]
\foreach \x in {2,4.5,6,7}
{	\draw (\x,0) -- (\x,8);
	\draw (0,\x) -- (8,\x);
}
\end{scope}

\begin{scope}[ultra thick]
\foreach \a/\b in {2/4.5,6/7}
{
	\draw (\a,0) -- (\b,0);
	\draw (0,\a) -- (0,\b);
}
\end{scope}

\foreach \a/\b/\c/\d in {2/2/4.5/4.5,2/6/4.5/7,6/2/7/4.5,6/6/7/7}
	\filldraw[fill=mycolor, draw=black] (\a,\b) rectangle (\c,\d);

}
\newcommand{\wssussupport}
{
\myaxes{8}
\myticks{2,5,6,7}
\begin{scope}[ultra thick]
\draw (0.5,0) -- (7,0);
\draw (0,0.5) -- (0,7);
\draw[mycolor,latex-latex] (0.5,0.5) -- (7,7);
\end{scope}
}

\newcommand{\curvysupport}
{
\myaxes{8}
\myticks{1,7}

\tikzstyle{every node}=[]; 
\filldraw[fill=mycolor, draw=black]
			(1,1) .. controls +(90:3) and +(200:1) .. (3,5) node {}
						  .. controls +(20:2) and +(250:2) .. (5,7) node {}
						  .. controls +(-60:1) and +(200:0.5) .. (7,7) node {}
						  .. controls +(-110:0.5) and +(150:1) .. (7,5) node {}
						  .. controls +(-160:2) and +(70:2) .. (5,3) node {}
						  .. controls +(-110:1) and +(0:3) .. (1,1); node {}
\begin{scope}[dashed]
\foreach \x in {1, 7}
{	\draw (\x,0) -- (\x,8);
	\draw (0,\x) -- (8,\x);
}
\end{scope}

\begin{scope}[ultra thick]
\foreach \a/\b in {1/7}
{
	\draw (\a,0) -- (\b,0);
	\draw (0,\a) -- (0,\b);
}
\end{scope}
}
\newcommand{\torusdraw}[3]
{	
	\tikzset{weak lines/.style={gray, very thin}} 
	\pgfmathparse{int(#3)}
	\let\W\pgfmathresult
	\def\A{7mm}
	\foreach \x in {0,...,\W} 
	{	\foreach \y in {0,...,\W}	
		{ 	\pgfmathsetmacro\xx{int(#1)} 
 			\pgfmathsetmacro\yy{int(#2)} 
 			\pgfmathsetmacro\opx{(\yy+0.5)/(\W+0.5)} 
 			\pgfmathsetmacro\botx{\yy*100/(\W+1)} 
 			\pgfmathsetmacro\topx{(\yy+1)*100/(\W+1)} 
 			\pgfmathsetmacro\opy{(\xx+0.5)/(\W+0.5)} 
			\ifthenelse{\equal{\xx}{0.0}}%
			{	
					\fill[color=mypurple, fill opacity=\opx] (\x*\A -\A/2,\y*\A-\A/2) rectangle +(\A,\A);
			}{};
			\ifthenelse{\equal{\yy}{0.0}}%
			{	\fill[color=myorange, fill opacity=\opy] (\x*\A-\A/2,\y*\A -\A/2) rectangle +(\A,\A);
			}{};
			\pgftext[at={\pgfpoint{\x*\A}{\y*\A}}]{\scriptsize\pgfmathprintnumber{\xx},\pgfmathprintnumber{\yy}};			
		}
	};
	\draw (-0.5,-0.5) rectangle (\W-0.5,\W-0.5);
	\useasboundingbox (-0.5,-0.5) rectangle (\W-0.5, \W-0.5);
}

\colorlet{mycolor}{blue!20}

\renewcommand\thesection{\arabic{section}}

\begin{document}

\title*{Cornerstones of Sampling of Operator Theory}
\author{David Walnut, G\"otz E. Pfander, Thomas Kailath}
\institute{David Walnut \at George Mason University, Fairfax, Virginia, USA \email{dwalnut@gmu.edu}
\and G\"otz E. Pfander \at Jacobs University, Bremen, Germany, \email{g.pfander@jacobs-university.de}
\and Thomas Kailath \at Stanford University, California, USA \email{kailath@stanford.edu}
\and The inversion of the traditional alphabetical ordering of authors is at the suggestion of the third author, who desires that those at the end of the alphabet get some recognition. }

%
%
\maketitle


\abstract{This paper reviews some results on the identifiability of classes of operators whose Kohn-Nirenberg
symbols are band-limited (called {\em band-limited operators}), which we refer to as {\em sampling of operators}.
We trace the motivation and history of the subject back to the original work of the third-named author in the late 1950s
and early 1960s, and to the innovations in spread-spectrum communications that preceded that work.
We give a brief overview of the NOMAC (Noise Modulation and Correlation) and Rake receivers, which were early
implementations of spread-spectrum multi-path wireless communication systems.
We examine in
detail the original proof of the third-named author characterizing identifiability of channels in terms of the maximum time and
Doppler spread of the channel, and do the same for the subsequent generalization of that work by Bello.
The mathematical limitations inherent in the proofs of Bello and the third author are removed by using mathematical tools
unavailable at the time.
We survey more recent advances in sampling of operators and discuss the implications of
the use of periodically-weighted delta-trains as identifiers for operator classes that
satisfy Bello's criterion for identifiability, leading to new insights into the theory of finite-dimensional Gabor systems.
We present  novel results on operator sampling
in higher dimensions, and review implications and generalizations
of the results to stochastic operators, MIMO systems, and operators with unknown spreading domains.}

\section{Introduction}

The problem of identification of a time-variant communication channel arose in the 1950s as the problem of secure long-range wireless communications became increasingly important due to the geopolitical situation at the time.  Some of the theoretical and practical advances made then are described in this paper, and more recent advances extending the theory to more general operators, and onto a more rigorous mathematical footing, known as {\em  sampling of operators} are developed here as well.

The launching point for the theory of operator sampling is the early work of the third-named author in his Master's thesis at MIT, entitled ``Sampling models for linear time-variant filters''
\cite{Kai59}, see also \cite{Kai62,Kai63}, and \cite{Kai61} in which he reviews the identification problem for time-variant channels.  The third named author as well as  Bello in subsequent work \cite{Bel69} were attempting to
understand and describe the theoretical limits of identifiability of time-variant communication channels.
Section~\ref{section:historical} of this paper describes in some detail their work
and explores some of the important mathematical challenges they faced.  In Section~\ref{section:operatorsampling}, we describe the more recently developed framework of operator sampling.  Results addressing the problem considered by Bello are based on insights on finite dimensional Gabor systems which are presented in Section~\ref{section:finiteGabor}. Malikiosis's recent result \cite{M13} allows for the generalization of those results to a higher-dimensional setting, these are stated and proven in Section~\ref{section:higherdimensional}. We conclude the paper in Section~\ref{section:outlook} with a short summary of the sampling of operators literature, that is, of results presented in detail elsewhere.

\section{Historical Remarks.}\label{section:historical}

\subsection{The Cold War Origins of the Rake System.}

In 1958, Price and Green published {\em A Communication Technique for Multi-path Channels} in Proc. IRE \cite{PG58},
in which they describe a communication system called Rake, designed to solve the {\em multi-path problem}. When a wireless
transmitter does not have line-of-sight with the receiver, the  transmitted
signal is reflected possibly multiple times before reaching the receiver. Reflection by  stationary objects such as the ground or  buildings introduces random time delays to the signal, and  reflection
or refraction by  moving objects such as clouds, the troposphere, ionosphere, or a moving vehicle produce random frequency or
Doppler shifts in the signal as well.
Due to scattering and absorption, the reflected signals are randomly amplitude-attenuated too.  The problem is to recover the transmitted signal
as accurately as possible from the
superposition of time-frequency-shifted and randomly amplitude-attenuated versions of it. Since the location and velocities of the reflecting objects
change with time, the effects of the unknown, time-variant channel must be estimated and compensated for.

Price and Green's paper \cite{PG58} was the disclosure in the literature of a long-distance
system of wide-band or spread-spectrum communications that had been developed in response to strategic needs
related to the Cold War.  This fascinating story has been described in several articles by those directly
involved (\cite{Sch82, Sch83, P83, G08}). We present a summary of those remarks and of the Rake system below.
The goal is to motivate the original work of the third-named author on which
the theory of operator sampling is based.

In the years following World War II, the Soviet Union was exercising its power
in Eastern Europe with a major point of contention being Berlin, which the Soviets blockaded in the late 1940s.  This made secure communication
with Berlin a top priority.  As Paul Green describes it,
\begin{quote}
[T]he Battle of Berlin was raging, the Russians having isolated the city physically on land, so that the Berlin Airlift was resorted to,
and nobody knew when all the communication links would begin to be subjected to heavy Soviet jamming. \cite{G08}
\end{quote}
By 1950, with a shooting war in Korea about to begin, the Army Signal Corps approached researchers at MIT to develop secure,
and reliable wireless communication with the opposite
ends of the world.  According to Green,
\begin{quote}
It is difficult today to recall the fearful excitement of those times.  The Russians were thought to be 12 feet high in anything having to do with
applying mathematics to communication problems (``all Russians were Komogorovs or Kotelnikovs'')....[T]here was a huge backlog of unexploited
theory lying around, and people were beginning to build digital equipment with the unheard of complexity of a hundred or so vacuum tube-based bits (!).
And the money flowed. \cite{G08}
\end{quote}
The effort was called Project Lincoln (precursor to Lincoln Laboratory).  The
researchers were confronted by two main problems:  1) making a communications system robust to noise and deliberate jamming, and
2) enabling good signal recovery from multiple paths.

\subsection{Spread Spectrum communications and NOMAC}

The technique chosen to address the first problem is an application of the notion, already well-understood and used by that time,
that combatting distortions from noise and jamming can
be achieved by spreading the signal over a wide frequency band.  The idea of spreading the spectrum had been around for a long time \cite{P83, T80,
PSM82} and can be found even in a now famous Hedy Lamarr-George Antheil patent of 1942 \cite{MA42, P83},
which introduced the concept later called ``frequency hopping''. The system called NOMAC (Noise Modulation and Correlation) was developed in the early 1950s and used noise like (pseudo-noise or PN) signals to achieve spectrum spreading. Detailed discussion of its history can be found in \cite{P83,G08,WW92}.

The huge backlog of ``unexploited theory'' mentioned above included the recent work of Claude Shannon on communication theory \cite{Sha49}, of Norbert Wiener on correlation functions and least mean squares prediction and filtering \cite{W49}, and recent applications of statistical decision theory to detection problems in radar and communications.

The communication problem addressed by NOMAC was to encode data represented by a string of ones and zeros into analog signals that could be electromagnetically transmitted over a noisy communication channel in a way that foiled ``jamming'' by enemies.  The analog signals $x_1(\cdot)$ and $x_0(\cdot)$, commonly called Mark and Space, associated with the data digits 1 and 0, were chosen to be waveforms of approximate bandwidth B, and with small cross correlation.  The target application was $60$\,wpm teletype, with $22$ msec per digit (called a baud), which corresponds to a transmission rate of $1/0.022 \,\text{sec} =45$\,Hz. The transmitted signals were chosen to have a bandwidth of $10$\,KHz, which was therefore expected to yield a ``jamming suppression ration'' of $10,000/45=220$, or $23$\,db \cite{G08,WW92}. The jamming ratio is often called the ``correlation gain'', because the receiver structure involves cross correlation of the received signal with each of the possible transmitted signals.  If the correlation with the signal $x_1(\cdot)$ is larger than the one with the signal  $x_0(\cdot)$, then it is decided that the transmitted signal corresponded to the digit 1. This scheme can be shown to be optimum in the sense of minimum probability of error provided that the transmitted signals are not distorted by the communications channel and that the receiver noise is white Gaussian noise (see, for example, \cite{Hel60}). The protection against jamming is because unless the jammer has good knowledge of the noise like transmitted signals, any jamming signals would just appear as additional noise at the output of the crosscorrelations.

More details on the nontrivial ideas required for building a practical system can be found in the references.
We may mention that the key ideas arose from three classified MIT dissertations by Basore \cite{B52}, Pankowski \cite{P52}, and Green \cite{G53},
in fact, documents on NOMAC remained classified until 1961 \cite{G08}.

A transcontinental experiment was run on a NOMAC system, but was found to have very poor performance because of the presence of multiple paths;  the signals arriving at the receiver by these different paths sometimes interfere destructively. This is the phenomenon of ``fading'', which causes self jamming of the system.  Some improvement was achieved by adding additional circuitry and the receiver to separately identify and track the two strongest signals and combine them after phase correction; this use of time and space diversity enabled a correlation gain of $17$\,db, $6$\,db short of the expected performance. It was determined that this loss was because of the neglected weaker paths, of which there could be as many as 20 or 30. So attention turned to a system that would allow the use of all the different paths.

\subsection{The Rake system}

One conceptual basis for this new system was provided by the doctoral thesis of Robert Price \cite{P53}, the main results of which were published in 1956 \cite{P56}. In a channel with severe multi-path the signal at the receiver is composed of  large number of signals of different amplitudes and phases and so Price made the assumption that the received ``signal'' was a Gaussian random process.  He studied the problem of choosing between the hypothesis
$$
	H_i: w(\cdot)=Ax_i(\cdot)+n(\cdot), \quad i=0,1,
$$
where the random time variant linear communication channel $A$ is such that the  $\{Ax_i(\cdot)\}$ are Gaussian processes.  In this case, the earlier cross correlation detection scheme makes no sense, because the ``signal'' arriving at the receiver is not deterministic but is a sample function of a random process, which is not available to the receiver because it is corrupted by the additive noise.  Price worked out the optimum detection scheme and then ingeniously interpreted the mathematical formulas to conclude that the new receiver forms least mean-square estimates of the $\{Ax_i(\cdot)\}$ and then crosscorrelates the $w(\cdot)$ against these estimates.  In practice of course, one does not have enough statistical information to form these estimates and therefore more heuristic estimates are used and this was done in the actual system that was built.  The main heuristic, from Wiener's least mean-square smoothing filter solution and earlier insights, is that one should give greater weight to paths with higher signal-to-noise ratio.

So Price and Green devised a new receiver structure comprised of a delay line of length $3$\,ms intervals (the maximum expected time spread in their channel), with 30 taps spaced every $1 / 10$\,Khz, or $100\, \mu$s. This would enable the capture of all the multi-path signals in the channel.  Then the tap gains were made proportional to the strength of the signal received at that tap.  Since a Mark/Space decision was only needed every $22$\,ms (for the transmission rate of $60$\,wpm), and since the fading rate of the channel was slow enough that the channel characteristics remain constant over even longer than $22$\,ms, tap gains could be averaged over several $3$\,ms intervals.
 The new system was called  ``Rake'', because the delay line structure resembled that in a typical garden rake!

 Trials showed that this scheme worked well enough to recover the $6$\,db loss experienced by the NOMAC system.  The system was put into production and was successfully used for jam-proof communications between Washington DC and Berlin during the ``Berlin crisis'' in the early 60s.

 HF communications is no longer very significant, but the Rake receiver has found application in a variety of problems such as sonar, the detection of underground nuclear explosions, and planetary radar astronomy (pioneered by Price and Green, \cite{G68,P68}) and currently it is much used in mobile wireless communications.  It is interesting to note that the eight racks of equipment needed to build the Rake system in the 1960s is now captured in a small integrated circuit chip in a smart phone!

 However the fact that the Rake system did not perform satisfactorily when the fading rates of the communication channel were not very slow led MIT professor John Wozencraft, (who had been part of the Rake project team at Lincoln Lab) to suggest in 1957 (even before the open 1958 publication of the Rake system) to his new graduate student Thomas Kailath a fundamental study of linear time-variant communication channels and their identifiability for his Masters thesis.  While linear time-variant linear systems had begun to be studied at least as early as 1950 (notably by Zadeh \cite{Z50}), in communication systems there are certain additional constraints, notably limits on the bandwidths of the input signal and the duration of the channel memory. So a more detailed study was deemed to be worthwhile.

\subsection{Kailath's Time-Variant Channel Identification Condition}\label{section:KailathSufficient}

In the paper \cite{Kai59}, the author considers the problem of measuring a channel whose characteristics vary rapidly with time.
He considers the dependence of any theoretical channel estimation scheme on how rapidly
the channel characteristics change and concludes that there are theoretical limits on the ability
to identify a rapidly changing channel.  He models the channel $A$ as a linear time-variant filter and defines

\begin{quote}
$A(\lambda,t)=$ response of $A$, measured at time $t$ to a unit impulse input at time $t-\lambda$. 
\end{quote}

$A(\lambda,t)$ is one form of the time-variant impulse response of the linear channel that emphasizes the role of the ``age'' variable $\lambda$. The channel response to an input signal $x(\cdot)$ is
$$Ax(t) = \int A(\lambda,t)\,x(t-\lambda)\,d\lambda.$$
An impulse response $A(\lambda,t)=A(\lambda)$ represents a time-invariant filter.  Further, the author states

\begin{quote}
Therefore the rate of variation of $A(\lambda,t)$ with $t$, for fixed $\lambda$, is a measure of the rate of variation of the filter.
It is convenient to measure this variation in the frequency domain by defining a function $\mathcal A$
$$ \mathcal A(\lambda, f)=\int_{-\infty}^\infty  A(\lambda,t) e^{-2\pi i f t } dt \quad
$$
\end{quote}

Then he defines
$$ B=\max_{\lambda} [ b-a, \text{ where } \mathcal A (\lambda,f)=0 \text{ for } f\notin [a,b] \,].$$
While symmetric support is assumed in the paper, this definition makes clear that non-rectangular regions of support are already in view.
Additionally, he defines the memory as the maximum time-delay spread in response to an impulse of the channel as
$$L=\max_{t} [\min_{\lambda'}\text{ such that }  A (\lambda,t)=0, \ \lambda\geq \lambda' ].$$
In short, the assumption in the continuation of the paper is that
$$\supp\mathcal A (\lambda,f) \subseteq [0,L]\times[-W,W]$$
where $W=B/2$.  The function $\mathcal A (\lambda,f)$ is often called the \emph{spreading function} of the channel.
He then asks under what assumptions on $L$ and $B=2W$ can such a channel be measured?  In the context of the Rake system, 
this translates to the question of whether there are limits on the rate of variation of the filter that can assure that the measurement filter
can be presumed to be effective.

The author's assertion is that as long as $BL\le 1$, then a ``simple measurement scheme'' is sufficient.

\begin{quote}
We have assumed that the bandwidth of any ``tap function'', $A_\lambda(\cdot)\,[=A(\lambda,\cdot)]$ , is limited to a frequency region of width $B$,
say a low-pass region $(-W,W)$ for which $B=2W$. Such band-limited taps are determined according to the Sampling theorem, by their values
at the instants $i/2W$, $i=0,\pm 1, \pm 2, \ldots $.

If the memory, $L$, of the filter, $A(\lambda,t)$ is less than $1/2W$ these values are easily determined: we put in unit impulses to
$A(\lambda,t)$ at instants $0,\ 1/2W,\  2/2W, \ldots, T$, and read off from the responses the desired values of the impulse response $A(\lambda,t)$. [...] If $L\leq 1/2W$, the responses to the different input impulses do not interfere with one another and the above values can be unambiguously determined.
\end{quote}

In other words, sufficiently dense samples of the tap functions  can be obtained by sending an impulse train $\sum_n \delta_{n/2W}$ through the channel. Indeed,
$$A\big(\sum_n \delta_{n/2W}\big)(t) = \sum_n \int A(\lambda,t)\,\delta_{n/2W}(t-\lambda)\,d\lambda = \sum_n A(t-n/2W,t).$$
Evaluating the operator response at  $t=\lambda_0+n_0/2W$, $n_0\in\Z$, we obtain
\begin{eqnarray*}
A\big(\sum_n \delta_{n/2W}\big)(\lambda_0+n_0/2W)
&  =  &  \sum_n A(\lambda_0+(n_0-n)/2W,\lambda_0+n_0/2W) \\
&  =  &  A(\lambda_0,\lambda_0+n_0/2W)
\end{eqnarray*}
since $L\le 1/2W$ implies that $A(\lambda_0+(n_0-n)/2W,\lambda_0+n_0/2W)=0$ if $n\neq n_0$. In short,  for each $\lambda$, the samples $A(\lambda,\lambda+n/2W)$
for $n\in\Z$ can be recovered.

The described Kailath sounding procedure is depicted in Figure~\ref{fig:KailathSounding}.  In this visualization, we plot the kernel $\kappa(s,t)=A(t-s,t)$ of the operator $A$, that is,
$$Ax(t) = \int A(\lambda,t)\,x(t-\lambda)\,d\lambda = \int A(t-s,t)\,x(s)\,ds = \int \kappa (t,s)\,x(s)\,ds.$$

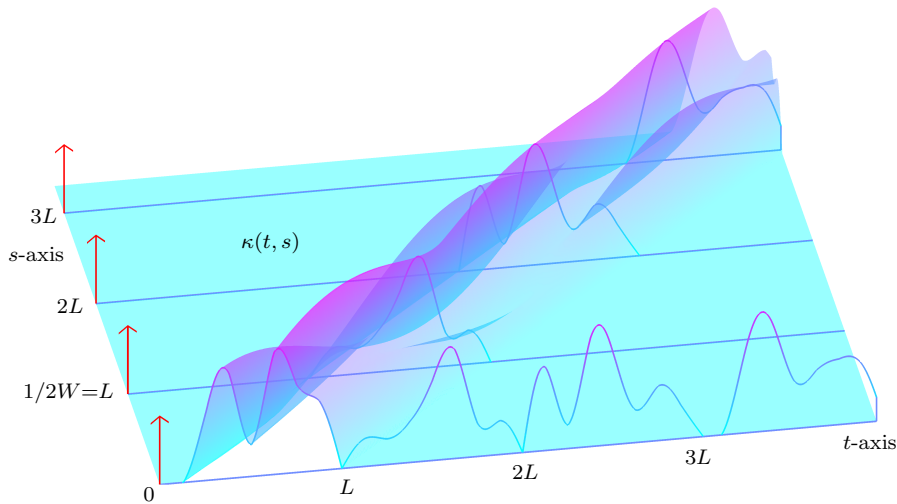
\begin{figure}\label{fig:KailathSounding}
\hspace{-.2cm}\begin{tikzpicture}
\node[above right]  at (0,1){
\includegraphics[width=11cm]{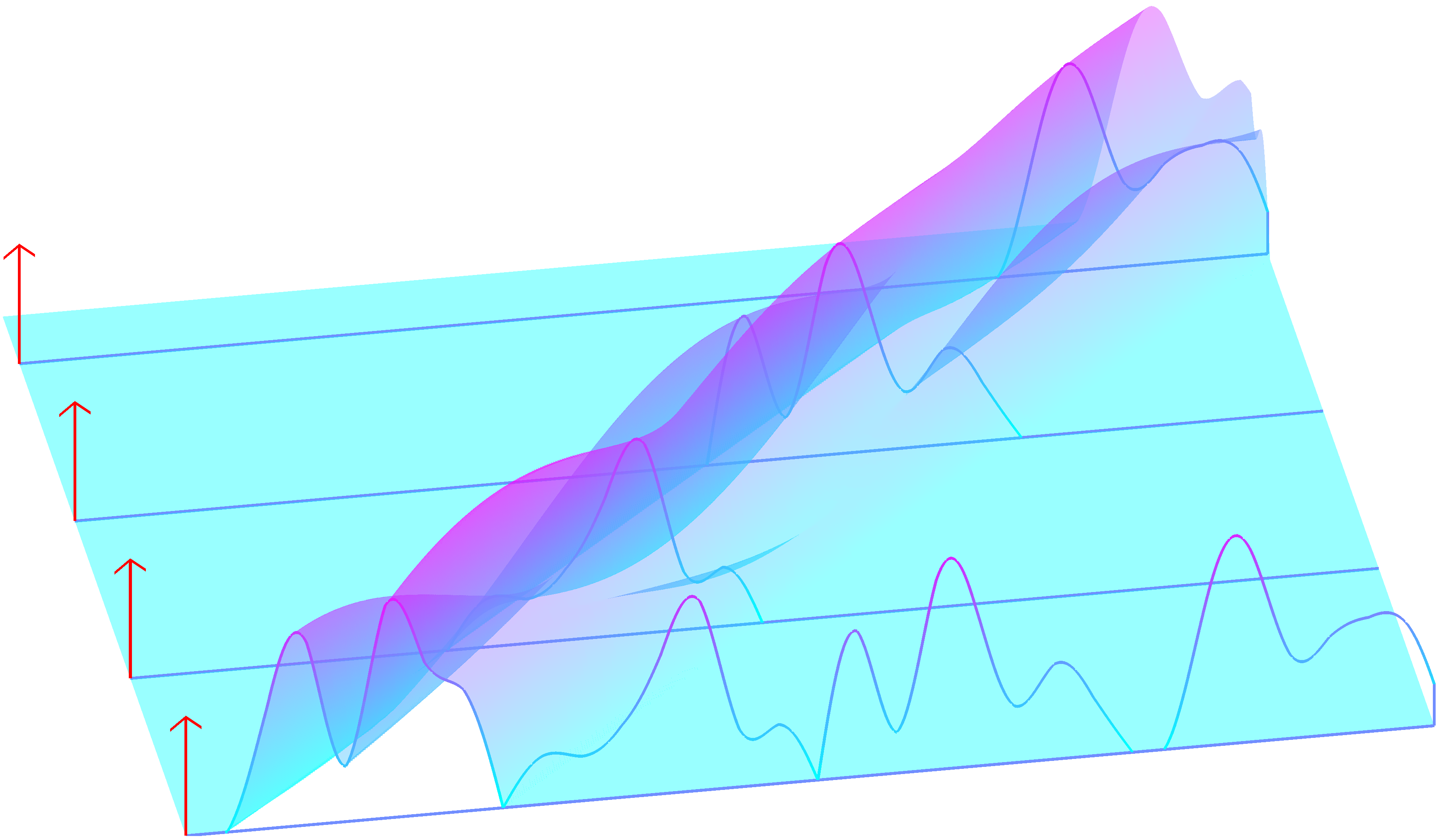}
};
\node[above right]  at (10.5,1.5) {$t$-axis};
\node[above right] at (-.6,4){$s$-axis};
\node[above right]  at (2.5,4.1){$\kappa(t,s)$};
\node[above right]  at (1.2,.8){$0$};
\node[above right]  at (-0.4,2.1){$1/2W{=}L$};
\node[above right]  at (0.05,3.3){$2L$};
\node[above right]  at (-.3,4.5){$3L$};
\node[above right]  at (3.8,.9){$L$};
\node[above right]  at (6.1,1.1){$2L$};
\node[above right]  at (8.4,1.3){$3L$};
\end{tikzpicture}

\caption{Kailath sounding of $A$ with $\supp\mathcal A (\lambda,f) \subseteq [0,L]\times[-W,W]$ and $L=1/2W$. The kernel $\kappa(t,s)$ is displayed on the $(t,s)$ plane, the impulse train $\sum_n \delta_{n/2W}(s)$ on the $s$-axis, and the output signal $Ax(t)=A\big(\sum_n \delta_{n/2W}\big)(t)
=\sum_n A(t-n/2W,t)=\sum_n\kappa (t,n/2W)$. The sample values of the tab functions $A_\lambda(t)=A(\lambda,t)=\kappa(t,t-\lambda)$ can be read off $Ax(t)$.}
\label{fig:KailathSounding}
\end{figure}

%

\subsection{Necessity of Kailath's Condition for Channel Identification.}\label{section:kailathnecessity}

For the ``simple measurement scheme'' to work, $BL\le 1$ is sufficient but could be restrictive.

\begin{quote}
We need, therefore, to devise more sophisticated measurement schemes.  However, we have not pursued this question very far because for
a certain class of channels we can show that the condition
$$
	L\leq 1/2W, \text{ i.e. }, BL\leq 1
$$
is necessary as well as sufficient for unambiguous measurement of $A(\lambda,t)$. The class of channels is obtained as follows:
We first assume that there is a bandwidth constraint on the possible input signals to $A(\lambda,t)$, in that the signals are
restricted to $(-W_i,W_i)$ in frequency. We can now determine a filter $A_{W_i}(\lambda,t)$ that is equivalent to $A(\lambda,t)$
over the bandwidth $(-W_i,W_i)$, and find necessary and sufficient conditions for unambiguous measurement of $A_{W_i}(\lambda,t)$.
If we now let $W_i\to \infty$, this condition reduces to condition (1), viz: $L\leq 1/2W$.
Therefore, condition (1) 
is valid for all filters $A(\lambda,t)$ that may be obtained as the limit of band-limited channels. This class includes almost
all filters of physical interest.  The argument is worked out in detail in Ref.~6 \footnote{Ref. 6 is \cite{Kai59}.} but we give a brief outline here.
\end{quote}

The class of operators in view here can be described as limits (in some unspecified sense) of operators whose
impulse response $A(\lambda,t)$ is bandlimited to $[-W_i,W_i]$ in $\lambda$ for each $t$ and periodic with period $T>0$ in $t$ for each
$\lambda$.
Here, $T$ is assumed to have some value larger than the maximum time over which the channel will be operated.  We could take it as the duration of the input signal to the channel.

The restriction to input signals bandlimited $(-W_i,W_i)$ indicates that it suffices to know the values of $A(\lambda,t)$ or ${\cal A}(\lambda, f)$ for a finite set of values of $\lambda$: $\lambda=0$, $1/2W_i$, $2/2W_i$, $\ldots$, $L$, assuming for simplicity that $L$ is a multiple of $1/2W_i$. Therefore, we can write
\begin{align*}
  A(\lambda,t) = \sum_n A(n/2W_i,t)\,\sinc_{W_i}(\lambda-n/2W_i),
\end{align*}
where  $\sinc_{W_i}(t) = \sin(2\pi W_i t)/(2\pi W_i t)$ so that as $W_i\to\infty$, $\sinc_{W_i}(t)$
becomes more concentrated at the origin.

Also, $T$-periodicity in $t$ allows us to write
$$A(\lambda,t) = \sum_k A(\lambda,k/T)\,e^{2\pi ikt/T},$$
so that combining gives
$$A(\lambda,t) = \sum_n \sum_k A(n/2W_i,k/T)\,\sinc_{W_i}(\lambda-n/2W_i)\,e^{2\pi ikt/T}.$$

Based on the restriction to bandlimited input signals which are $T$ periodic, we have obtained a representation of $A$ which is neither compactly supported in $\lambda$ nor bandlimited in $t$.
However, the original restriction that
$$\supp\mathcal A (\lambda,f) \subseteq [0,L]\times[-W,W]$$
 motivates the assumption that we are working with finite sums, viz.
$$A(\lambda,t) = \sum_{n/2W_i\in[0,L]} \sum_{k/t\in[-W,W]} A(n/2W_i,k/T)\,\sinc_{W_i}(\lambda-n/2W_i)\,e^{2\pi ikt/T}.$$
This is how the author obtains the estimate that there are at most $(2W_iL+1)(2WT+1)$ degrees of freedom in any impulse response $A$ in the given class.
%

For any input signal $x(t)$  bandlimited to $[-W_i,W_i]$,  the output will be bandlimited to $[-W-W_i,W+W_i]$.  Specifically,
\begin{eqnarray*}
Ax(t) &  =  &  \int A(\lambda,t)\,x(t-\lambda)\,d\lambda \\
      &  =  &  \sum_{n/2W_i\in[0,L]} \sum_{k/t\in[-W,W]} A(n/2W_i,k/T)\,e^{2\pi ikt/T} \\
      &     &  \qquad\qquad\qquad\int x(t-\lambda)\,\sinc_{W_i}(\lambda-n/2W_i)\,d\lambda \\
      &  =  &  \sum_{n/2W_i\in[0,L]} \sum_{k/t\in[-W,W]} A(n/2W_i,k/T)\,e^{2\pi ikt/T} \\
      &     &  \qquad\qquad\qquad(x\ast \sinc_{W_i})(t-n/2W_i).
\end{eqnarray*}
Since $e^{2\pi ikt/T}\,(x\ast \sinc_{W_i})(t-n/2W_i)$ is bandlimited to $[-W_i,W_i]+(k/T)$ for $k/T\in[-W,W]$, it follows that $Ax(t)$
is bandlimited to $[-W-W_i,W+W_i]$.

If we restrict our attention to signals $x(t)$ time-limited to $[0,T]$, the output signal $Ax(t)$  will have duration $T+L$, and $Ax(\cdot)$ will be completely determined by
its samples at $\frac{n}{2(W+W_i)}\in[0,T+L]$, from which we can identify $2(T+L)(W+W_i)+1$ degrees of freedom.

In order for identification to be possible, the number of degrees of freedom of the output signal must be at least as large as
the number of degrees of freedom of the operator, i.e.
\begin{align*}
 2W_iT+2W_i L+2W T+2W L+1&=\\2(T+L)(W_i+W)+1 &\geq  (2 W T + 1 ) (2 W_i L + 1) \\ &=2WT+2W_i L+1+4W_i W TL
\end{align*}
which reduces ultimately to
\begin{align*}
  \frac{1}{1-1/(2W_iT)}  \geq 2W L = B L.
\end{align*}
That is, $BL$ needs to be strictly smaller than $1$ in the approximation  while $BL=1$ may work in the limiting case $W_i \to \infty$
(and/or $T \to \infty$).

This result got a lot of attention because it corresponded with experimental evidence that Rake did not function well when the condition $BL<1$ was violated.  It led to the designation of ``underspread'' and ``overspread'' channels for which $BL$ was less than or greater than 1.

\subsection{Some Remarks on Kailath's Results}

This simple argument is surprising, particularly in light of the fact that the author obtained a deep result in time-frequency analysis
with none of the tools of modern time-frequency analysis at his disposal.  He very deftly uses the extremely useful engineering ``fiction''
that the dimension of the space of signals essentially bandlimited to $[-W,W]$ and time-limited to $[0,T]$ is approximately $2WT$.
The then recent papers of Landau, Slepian and Pollak \cite{SP61, LSP61}, which are mentioned explicitly in \cite{Kai59}, provided a rigorous
mathematical framework for understanding the phenomenon of essentially simultaneous band- and time-limiting.  While the existence of these
results lent considerable mathematical heft to the argument, they were not incorporated into a fully airtight mathematical
proof of his theorem.

\begin{quote}
In the proof we have used a degrees-of-freedom argument based on the sampling theorem which assumes strictly bandlimited functions.
This is an unrealistic assumption for physical processes.   It is more reasonable to call a process band (or time) limited if some
large fraction of its energy, say 95\%, is contained within a finite frequency (or time) region. Recent work by Landau and Slepian
has shown the concept of approximately $2TW$ degrees of freedom holds even in such cases.  This leads us to believe that our proof
of the necessity of the $BL\leq 1$ condition is not merely a consequence of the special properties of strictly band-limited functions.
It would be valuable to find an alternative method of proof.
\end{quote}

While Kailath's Theorem is stated for channel operators whose spreading functions are supported in a rectangle, it is clear that
the later work of Bello \cite{Bel69} was anticipated and more general regions were in view.  This is stated explicitly.

\begin{quote}
We have not discussed how the bandwidth, B is to be defined.  There are several possibilities: we might take the nonzero
$f$-region of $\mathcal A(\lambda,f)$; or use a``counting" argument. We could proceed similarly for the definition of $L$.
As a result of these several possibilities, the value 1, of the threshold in the condition $BL\leq 1$ should be considered
only as an order of magnitude value.

...constant and predictable variations in $B$ and $L$, due for example to known Doppler shifts or time displacements,
would yield large values for the absolute values of the time and frequency spreadings.
However such predictable variations should be subtracted out before the
$BL$ product is computed; {\em what appears to be important is the area covered in the time- and frequency-spreading
plane rather than the absolute values of $B$ and $L$.} (emphasis added)
\end{quote}


The reference to ``counting'' as a definition of bandwidth clearly indicates that essentially arbitrary regions of support for the operator
spreading function were in view here, and that a necessity argument relying on degrees of freedom and not the shape of the spreading
region was anticipated.  The third-named author did not pursue the measurement problem studied in his MS thesis because
he went on in his PhD dissertation to study the
optimum (in the sense of minimum probability of error) detector scheme of which Rake is an intelligent engineering approximation.  See
\cite{Kai60,Kai61,Kai63}.

The mathematical limitations of the necessity proof in \cite{Kai59} 
can be removed by addressing the identification problem directly as
a problem on infinite-dimensional space rather than relying on finite-dimensional approximations to the channel.  This approach also
avoids the problem of dealing with simultaneously time and frequency-limited functions.  In this way, the proof can be made completely
mathematically rigorous.  This approach is described in Section~\ref{section:pfandernecessity}.

\subsection{Bello's  time-variant Channel Identification Condition}

Kailath's Theorem was generalized by Bello in \cite{Bel69} along the lines anticipated in \cite{Kai59}.  Bello's argument follows that of \cite{Kai59}
in its broad
outlines but with some significant differences.  Bello clearly anticipates some of the technical difficulties that have been solved more
recently by the authors and others and which have led to the general theory of operator sampling. 

Continuing with the notation of this section, Bello considers channels with spreading function $\mathcal A (\lambda,f)$ supported in a
rectangle $[0,L]\times[-W,W]$.  If $L$ and $W$ are all that is known about the channel, then Kailath's criterion for measurability requires that
$2WL\le 1$.  Bello considers channels for which $2WL$ may be greater than $1$ but for which
$$S_A = |\supp\mathcal A (\lambda,f)| \le 1$$
and argues that this is the most appropriate criterion to assess measurability of the channel modeled by $A$.

In order to describe Bello's proof we will fix parameters $T\gg L$ and $W_i\gg W$ and following the assumptions earlier in this section, assume that
inputs to the channel are time-limited to $[0,T]$ and (approximately) bandlimited to $[-W_i,W_i]$.  Under this assumption, Bello considers
the spreading function of the channel to be approximated by a superposition of point scatterers, viz.
$$\mathcal A (\lambda,f) = \sum_n\sum_k A_{n,k}\,\delta(f-(k/T))\,\delta(\lambda-(n/2W_i)).$$
Hence the response of the channel to an input $x(\cdot)$ is given by
\begin{eqnarray}\label{eqn:belloresponse}
Ax(t) &  =  & \int\!\!\!\int x(t-\lambda)\,e^{2\pi i f(t-\lambda)}\mathcal A (\lambda,f)\,d\lambda\,df \\ \notag
&  =  &  \sum_n\sum_k A_{n,k}\,x(t-(n/2W_i))\,e^{2\pi i(k/T)(t-(n/2W_i))}.
\end{eqnarray}
Note that this is a continuous-time Gabor expansion with window function $x(\cdot)$ (see, e.g., \cite{Gro01}).  By standard density results
in Gabor theory, the collection of functions $\set{x(t-(n/2W_i))\,e^{2\pi i(k/T)(t-(n/2W_i))}}$ is overcomplete
as soon as $2TW_i>1$.  Consequently, without further discretization, the coefficients $A_{n,k}$ are
in principle unrecoverable.
Taking into consideration support constraints on $\mathcal A$, we assume that the sums are finite, viz.
$$\bigg(\frac{n}{2W_i},\frac{k}{T}\bigg) \in \supp\mathcal A.$$
Hence determining the channel characteristics amounts to finding $A_{n,k}$ for those pairs $(n,k)$.  It should be noted that for a given
spreading function $\mathcal A (\lambda,f)$ for which $\supp\mathcal A$ is a Lebesgue measurable set, given $\epsilon>0$, there exist
$T$ and $W_i$ sufficiently large that the number of such $(n,k)$ is no more than $2 S_A W_i T(1+\epsilon)$.  On the other hand,
for a given $T$ and $W_i$, there exist spreading functions $\mathcal A (\lambda,f)$ with arbitrarily small non-convex $S_A$ for which the number of
nonzero coefficients $A_{n,k}$ can be large.  For example, given $T$ and $W_i$, $S_A$ could consist of rectangles centered on the points
$(n/(2W_i), k/T)$ with arbitrarily small total area.

By sampling, (\ref{eqn:belloresponse}) reduces to a discrete, bi-infinite linear system, viz.
\begin{equation}\label{eqn:belloresponsediscrete}
Ax\bigg(\frac{p}{2W_i}\bigg) = \sum_n \sum_k A_{n,k}\,x\bigg(\frac{p-n}{2W_i}\bigg)\,e^{2\pi i\frac{k}{T}(\frac{p-n}{2W_i})}
\end{equation}
for $p\in\Z$.  Note that (\ref{eqn:belloresponsediscrete}) is the expansion of a vector in a discrete Gabor system on $\ell^2(\Z)$, 
a fact not mentioned by Bello, and of which he was apparently unaware. Specifically,
defining the translation operator $\mathcal T$ and the modulation operator $\mathcal M$ on $\ell^2$ by
\begin{equation}\label{eqn:translationandmodulation}
\mathcal Tx(n) = x(n-1),\qquad{\mbox{\rm and}}\qquad \mathcal Mx(n) = e^{\pi i n/(TW_i)}x(n),
\end{equation}
(\ref{eqn:belloresponsediscrete}) can be rewritten as
\begin{equation}\label{eqn:belloresponsediscretegabor}
Ax\bigg(\frac{p}{2W_i}\bigg) = \sum_n \sum_k (\mathcal T^n\,\mathcal M^k x)(p)\,A_{n,k}.
\end{equation}
Since there are only finitely many nonzero unknowns in this system, Bello's analysis proceeds by looking at finite sections of
(\ref{eqn:belloresponsediscretegabor}) and counting degrees of freedom.

\smallskip\noindent{\em Necessity.}  Following the lines of the necessity argument in \cite{Kai59}, we note that there are at least
$2(T+L)(W+W_i)$ degrees of freedom in the output vector $Ax(t)$, that is, at least that many independent samples of the form
$Ax(p/2W_i)$, and as observed above, no more than $2 S_A W_i T(1+\epsilon)$ nonzero unknowns $A_{n,k}$.  Therefore, in order for the $A_{n,k}$
to be determined in principle, it must be true that
$$2 W_i T(1+\epsilon) S_A \le 2(T+L)(W+W_i)$$
or
$$S_A \le \frac{(T+L)(W+W_i)}{W_i T(1+\epsilon)}.$$
Letting $T,\,W_i\to\infty$ and $\epsilon\to 0$, we arrive at $S_A\le 1$.

\smallskip\noindent{\em Sufficiency.}  Considering a section of the system (\ref{eqn:belloresponsediscretegabor}) based on the assumption
that $\supp\mathcal A \subseteq [0,L]\times[-W,W]$, the system has approximately $2W_i(T+L)$ equations in $(2W_iT)(2WL)$ unknowns.  Since $L$ and $2W$
are simply the dimensions of a rectangle that encloses the support of $\mathcal A$, $2WL$ may be quite large and independent of
$S_A$.  Hence the system will not in general be solvable.  However by assuming that $S_A<1$, only approximately
$S_A(2W_iT)$ of the $A_{n,k}$ do not vanish and the system reduces to one in which the number of equations is roughly equal to the
number of unknowns.  In this case it would be possible to solve (\ref{eqn:belloresponsediscretegabor}) as long as the collection of
appropriately truncated vectors $\{\mathcal T^n \mathcal M^k x\colon A_{n,k}\ne 0\}$ forms a linearly independent set for some vector $x$.

In his paper, Bello was dealing with independence properties of discrete Gabor systems apparently without realizing it, or at least without
stating it explicitly.
Indeed, he argues in several different ways that a vector $x$ that produces a linearly independent set should exist,
and intriguingly suggests that a vector consisting of
$\pm 1$ should exist with the property that the Grammian of the Gabor matrix corresponding to the section of (\ref{eqn:belloresponsediscretegabor})
being considered is diagonally dominant. 

The setup chosen below to prove Bello's assertion leads to the consideration of a matrix whose columns stem from a
Gabor system on a finite-dimensional space, not on a sequence space.

\section{Operator Sampling}\label{section:operatorsampling}

The first key contribution of operator sampling is the use of frame theory and time-frequency analysis to remove assumptions of simultaneous
band- and time-limiting, and also to deal with the infinite number of degrees of freedom in a functional analytic setting 
(Section~\ref{section:Operator classes and operator identification}).
A second key insight is the development of a ``simple measurement scheme'' of the type used by the third-named author but that allows for the
difficulties identified by Bello to be resolved.  This insight is the use of periodically-weighted delta-trains as measurement functions for a channel.
Such measurement functions have three distinct advantages.

First, they allow for the channel model to be essentially arbitrary and clarify the reduction of the operator identification problem to a
finite-dimensional setting without imposing a finite dimensional model that approximates the channel.  Second, it combines the naturalness of the
simple measurement scheme described earlier with the flexibility of Bello's idea for measuring channels with arbitrary spreading support.
Third, it establishes a connection between identification of channels and finite-dimensional Gabor systems and allows us to determine windowing
vectors with appropriate independence properties.

In Section~\ref{section:Operator classes and operator identification}, we introduce some operator-theoretic descriptions of some of the operator
classes that we are able to identify, and discuss briefly different ways of representing such operators.  Such a discussion is beneficial in several ways.
First, it contains a precise definition of identifiability, which comes into play when considering the generalization of the necessity
condition for so-called overspread channels (Section~\ref{section:pfandernecessity}).  Second, we can extend the necessity condition to a very large
class of inputs.  In other words, we can assert that in a very general sense, no input can identify an overspread channel. Third, it allows us to include
both convolution operators and multiplication operators (for which the spreading functions are distributions) in the operator sampling theory.
The identification of multiplication operators via operator sampling reduces to the classical sampling formula, thereby showing that classical sampling
is a special case of operator sampling.  In Section~\ref{section:pfandernecessity} we present a natural formalization of the original necessity proof of \cite{Kai59}
(Section~\ref{section:kailathnecessity}) to the infinite-dimensional setting, which involves an interpretation of the notion of an
under-determined system to that setting.  Finally, in Section~\ref{section:mainidentification} we present the scheme given first in \cite{PW06b, PW13} for the
identification of operator classes using periodically-weighted delta trains and techniques from modern time-frequency analysis.

\subsection{Operator classes and operator identification}\label{section:Operator classes and operator identification}

We formally consider an arbitrary operator as a {\em pseudodifferential operator} represented by
\begin{align}\label{eqn:operator1}
 Hf(x) = \int \sigma_H(x,\xi)\widehat{f}(\xi)\,e^{2\pi ix\xi}\,d\xi,
\end{align}
where $\sigma_H(x,\xi)\in L^2(\R^2)$ is the {\em Kohn-Nirenberg} (KN) symbol of $H$.  The {\em spreading function} $\eta_H(t,\nu)$
of the operator $H$ is the {\em symplectic Fourier transform} of the KN symbol, viz.
\begin{align}\label{eqn:operator2}\eta_H(t,\nu) = \int\!\!\!\!\int \sigma_H(x,\xi)\,e^{-2\pi i(\nu x - \xi t)}\,dx\,d\xi\end{align}
and we have the representation
\begin{align}\label{eqn:operator3}Hf(x) =  \int\!\!\!\!\int \eta_H(t,\nu)\,\mathcal{T}_t\,\mathcal{M}_\nu f(x)\,d\nu\,dt\end{align}
where $\mathcal{T}_tf(x) = f(x-t)$ is the {\em time-shift operator} and $\mathcal{M}_\nu f(x) = e^{2\pi i \nu x}\,f(x)$
is the {\em frequency-shift operator}.

This is identical to the representation given in \cite{Kai59} where $\eta_H(t,\nu) = \mathcal A(\nu,t)$, see Section~\ref{section:KailathSufficient}.

To see more clearly where the spreading function arises in the context of communication theory, we can define the {\em impulse response}
of the channel modeled by $H$, denoted $h_H(x,t)$, by
$$Hf(x) = \int h_H(x,t)\,f(x-t)\,dt.$$
Note that if $h_H$ were independent of $x$, then $H$ would be a convolution operator and hence a model for a time-invariant channel.
In fact, with $\kappa_H(x,t)$ being the {\em kernel} of the operator $H$,
\begin{align}
  Hf(x) &= \int \kappa_H(x,t)\,f(t)\,dt \\
  &= \int h_H(x,t)\,f(x-t)\,dt \\
  &= \iint \eta_H(t,\nu)\,e^{2\pi i\nu (x-t)}\,f(x-t)\,d\nu\,dt\label{eqn:operatorrepresentations1} \\
  &= \int \sigma_H(x,\xi)\, \widehat f(\xi)\,e^{2\pi i x \xi} d\xi,\label{eqn:operatorrepresentations2}
\end{align}
where
\begin{align}
h_H(x,t) &= \kappa_H(x,x-t) \nonumber  \\
    &= \int \sigma_H (x,\xi)\, e^{2\pi i \xi t}\, d\xi, \nonumber \\
    &= \int \eta_H(t,\nu)\, e^{2\pi i \nu (x-t)}\, d\nu. \label{eqn:symbolrelations}
\end{align}
With this interpretation, the maximum support of $\eta_H(t,\nu)$ in the first variable corresponds to the maximum spread of a delta impulse sent
through the channel and the maximum support of $\eta_H(t,\nu)$ in the second variable corresponds to the maximum spread of a pure frequency sent
through the channel.

Since we are interested in operators whose spreading functions have small support, it is natural to define the following operator classes,
called {\em operator Paley-Wiener spaces} (see \cite{Pfa10}).

\begin{definition}\label{generaloperatorpaleywienerspaces}
For $S\subseteq \R^2$, we define the operator Paley-Wiener spaces $OPW(S)$ by
\begin{align*}
OPW(S) & =  \{H\in \mathcal L (L^2(\R), L^2(\R)) \colon \  \supp\eta_H\subseteq S,\,\norm{\sigma_H}_{L^{2}}<\infty\}.
\end{align*}
\end{definition}

\begin{remark}\label{rem:generaloperatorpaleywienerspaces}
In \cite{Pfa10, PW06}, the spaces $OPW^{p,q}(S)$, $1\le p,\,q<\infty$, were considered, where  $L^2$-membership of $\sigma_H$ is replaced
$$\norm{\sigma_H}_{L^{p,q}} = \Big(\int\Big(\int\abs{\sigma_H (x,\xi)}^q d\xi\Big)^{p/q}\,dx\Big)^{1/p}$$
with the usual adjustments made when either $p=\infty$ or $q=\infty$.
$OPW^{p,q}(S)$ is a Banach space with respect to the norm $\norm{H}_{OPW^{p,q}} = \norm{\sigma_H}_{L^{p,q}}$.
Note that if $S$ is bounded, then $OPW^{\infty,\infty}(S)$  consists of all bounded operators whose spreading function is supported on $S$.  In fact, the operator norm is then equivalent to the $OPW^{\infty,\infty}(S)$ norm, where the constants depend  on $S$ \cite{KP12}.

The general definition is beneficial since it also allows the inclusion of convolution operators with kernels whose Fourier transforms
lie in $L^q(\R)$ ($OPW^{\infty,q}(\R)$) and multiplication operators whose multiplier is in $L^p(\R)$ ($OPW^{p,\infty}(\R)$).
\end{remark}

The goal of operator identification is to find an input signal $g$ such that each operator $H$ in a given class is completely and stably determined by $Hg$.
In other words, we ask that the operator $H\mapsto Hg$ be continuous and bounded below on its domain.  In our setting, this translates to the existence of $c_1,\,c_2 >0$ such that
\begin{equation}\label{eqn:boundedbelow}
c_1\,\|\sigma_H\|_{L^2} \le \|Hg\|_{L^2}\le c_2\,\|\sigma_H\|_{L^2},\quad H\in OPW(S).
\end{equation}
This definition of identifiability of operators originated in \cite{KP06}.
Note that (\ref{eqn:boundedbelow}) implies that the mapping $H\mapsto Hg$ is {\em injective}, that is, that $Hg=0$ implies that $H\equiv 0$, but is not
equivalent to it.  The inequality (\ref{eqn:boundedbelow}) adds to injectivity the assertion that $H$ is also stably determined by $Hg$ in the sense that
a small change in the output $Hg$ would correspond to a small change in the operator $H$.  Such stability is also necessary for the existence of an algorithm
that will reliably recover $H$ from $Hg$.
In this scheme, $g$ is referred to as an {\em identifier} for the operator class $OPW(S)$ and if (\ref{eqn:boundedbelow}) holds, we say
that {\em operator identification} is possible.

In trying to find an explicit expression for an identifier, we use as a starting point the ``simple measurement scheme'' of \cite{Kai59}, in which $g$
is a delta train, viz. $g = \sum_n\delta_{nT}$ for some $T>0$.  In the framework of operator identification the channel measurement criterion in \cite{Kai59} takes the
following form \cite{KP06, PW06b, Pfa10}.

\begin{theorem}\label{thm:main-simple}
  For $H\in OPW \big([0, T] {\times}[-  \Omega / 2, \Omega / 2]\big)$ with  $T\Omega{\leq} 1$, we have
  \begin{eqnarray}
    \|H\sum_{k\in\Z}\delta_{kT}\|_{L^2(\R)}=T\|\sigma_H\|_{L^2},\notag 
  \end{eqnarray} and  $H$ can be reconstructed by means of
  \begin{eqnarray}
    \kappa_H(x+t,x)=\chi_{[0,T]}(t)\sum_{n\in\Z} \big(H\sum_{k\in\Z}\delta_{kT}\big)(t+nT)\, \frac{\sin(\pi T (x-n))}{\pi T (x-n)} \label{eqn:operatorreconstruction-simple}
  \end{eqnarray}
  where $\chi_{[0,T]}(t)=1$ for $t\in[0,T]$ and $0$ elsewhere and with convergence in the $L^2$  norm
  and uniformly in $x$ for every $t$.
\end{theorem}

As was observed earlier, the key feature of this scheme is that the spacing of the deltas in the identifier is sufficiently large so as to allow
the response of the channel to a given delta to ``die out'' before the next delta is sent.  In other words, the parameter $T$ must exceed the time-spread
of the channel.  On the other hand, the rate of change of the channel, as measured by its bandwidth $\Omega$, must be small enough that its impulse
response can be recovered from ``samples'' of the channel taken $T$ time units apart.  In particular, the samples of the impulse response $T$ units
apart can be easily determined from the output.  In the general case considered by Bello, in which the spreading
support of the operator is not contained in a rectangle of unit area, this intuition breaks down.

Specifically, suppose that we consider the operator class $OPW(S)$ where $S\subseteq[0,T_0]\times[-\Omega_0/2,\Omega_0/2]$ and $T_0\Omega_0\gg 1$
but where $|S|<1$.
Then sounding the channel with a delta train of the form $g=\sum_n \delta_{nT_0}$ would severely {\em undersample} the impulse response function.
Simply increasing the sampling rate, however, would produce overlap in the responses of the channel to deltas close to each other.  An approach to
the undersampling problem in the literature of classical sampling theory is to sample at the low rate transformed versions of the function,
chosen so that the interference of the several undersampled functions can be dealt with.  This idea has its most classical expression
in the Generalized Sampling scheme of Papoulis \cite{Pa77}.  Choosing shifts and constant multiples of our delta train results in an identifier
of the form $g=\sum_n c_n\,\delta_{nT}$ where the weights $(c_n)$ have period $P$ (for some $P\in\N$) and $T>0$ satsifies $PT>T_0$.

If $g$ is discretely supported (for example, a periodically-weighted delta-train), then we refer to operator identification as {\em operator sampling}.
The utility of periodically-weighted delta trains for operator identification is a cornerstone of operator sampling and has far-reaching implications
culminating in the developments outlined in Sections~\ref{section:higherdimensional} and \ref{section:outlook}.


\subsection{Kailath's necessity proof and operator identification}\label{section:pfandernecessity}

In Section~\ref{section:kailathnecessity} we presented the proof of the necessity of the condition $BL\le 1$ for channel identification as given
in \cite{Kai59}.  The argument consisted of finding a finite-dimensional approximation of the channel $H$, and then showing that, given any putative
identifier $g$, the number of degrees of freedom present in the output $Hg$ must be at least as large as the number of degrees of freedom in the
channel itself.  For this to be true in any finite-dimensional setting, we must have $BL<1$ and so in the limit we require $BL\le 1$.
In essence, if $BL>1$, we have a linear system with fewer equations than unknowns which necessarily has a nontrivial nullspace.  The generalization of this notion to
the infinite-dimensional setting is the basis of the necessity proof that appears in \cite{KP06}.  In this section, we present an outline of that
proof, and show how the natural tool for this purpose once again comes from time-frequency analysis.

To see the idea of the proof, assume that $BL>1$ and for simplicity let
$S=[-\frac L 2,\frac L 2]\times[-\frac B 2,\frac B 2]$.  The goal is to show that for any sounding signal $s$ in an appropriately
large space of distributions\footnote{$S'_0(\R)$, the dual space of the Feichtinger algebra $S_0(\R)$
\cite{Gro01}, or ${\cal S}'(\R)$, the space of tempered distributions \cite{PW06}. These spaces are large enough to contain weighted infinite
sums of delta distributions.},
the operator $\Phi_s\colon OPW(S)\longrightarrow L^2(\R)$, $H\mapsto Hs$, is not stable, that is, it does not possess a lower bound
in the inequality (\ref{eqn:boundedbelow}).

First, define the operator $E\colon l_0(\Z^2)\longrightarrow OPW(S)$,
where $l_0(\Z^2)$ is the space of finite sequences equipped with the $l^2$ norm, by
$$E(\sigma) = E(\set{\sigma_{k,l}})
  = \sum_{k,l}\sigma_{k,l} \mathcal M_{k\lambda/L}\mathcal T_{l\lambda/B}\,P\,\mathcal T_{-l\lambda/B}\mathcal M_{-k\lambda/L}$$
where $1<\lambda$ is chosen so that $1<\lambda^4<BL$ and where $P$ is a time-frequency localization operator
whose spreading function $\eta_P(t,\nu)$ is infinitely differentiable, supported in $S$, and identically one on
$[-\frac L {2\lambda},\frac L {2\lambda}]\times[-\frac B {2\lambda},\frac B {2\lambda}]$.
It is easily seen that the operator $E$ is well-defined and has spreading function
$$\eta_{E(\sigma)}(t,\nu) = \eta_P(t,\nu)\,\sum_{k,l}\sigma_{k,l}\,e^{2\pi i(k\lambda t/L - l\lambda\nu/B)}.$$
By construction, it follows that for some constant $c_1$,
$\norm{E(\sigma)}_{OPW(S)}\ge c_1\norm{\sigma}_{l^2(\Z^2)}$, for all $\sigma$, and that for any distribution $s$, $Ps$ decays rapidly in time and in
frequency.

Next define the Gabor analysis operator $C_g\colon L^2(\R)\longrightarrow l^2(\Z^2)$ by
$$C_g(s) = \set{\ip{s}{\mathcal M_{k\lambda^2/L}\mathcal T_{l\lambda^2/B}g}}_{k,l\in\Z}$$
where $g(x)=e^{-\pi x^2}$.  A well-known theorem in Gabor theory asserts that
$\set{\mathcal M_{k\alpha}\mathcal T_{l\beta}g}_{k,l\in\Z}$ is a Gabor frame for $L^2(\R)$ for every
$\alpha\beta<1$ (\cite{Lyu92,SW92,Sei92b}).  Consequently $C_g$ satisfies, for some $c_2>0$, $\norm{C_g(s)}_{l^2(\Z^2)}\ge c_2\,\norm{s}_{L^2(\R)}$ for all $s$,
since $\lambda^2/L\,\cdot\lambda^2/B = \lambda^4/BL < 1$.

For any $s$, consider the composition operator
$$ C_g\circ\Phi_s\circ E \colon l_0(\Z^2) \longrightarrow l^2(\Z^2).$$
The crux of the proof lies in showing that this composition operator is not stable, that is, it does not have a lower bound.
Since $C_g$ and $E$ are both bounded below, it follows that $\Phi_s$ cannot be stable.
Since $s\in S'_0(\R)$ was arbitrary, this completes the proof.

To complete this final step we examine the canonical bi-infinite matrix representation of the
above defined composition of operators,
that is, the matrix $M=(m_{k',l',k,l})$ that satisfies
$$(C_g\circ\Phi_s\circ E(\sigma))_{k',l'} = \sum_{k,l} m_{k',l',k,l}\,\sigma_{k,l}.$$
It can be shown that $M$ has the property that for some rapidly decreasing function $w(x)$,
\begin{equation}\label{eqn:matrixdecay}
\abs{m_{k',l',k,l}} \le w(\max\set{\abs{\lambda k'-k},\abs{\lambda l'-l}}).
\end{equation}
The proof is completed by the following Lemma.  Its proof can be found
in \cite{KP06} and generalizations can be found in \cite{Pfa05}.

\begin{figure}[H]\label{fig:skewmatrix}
\centering
\begin{tikzpicture} 
\let\mymatrixcontent\empty %
\let\mymatrixrow\empty %
\foreach \x in {1,...,22}{ \expandafter\gappto\expandafter\mymatrixrow\expandafter{ \& }} %
\foreach \y in {1,...,12}{ \expandafter\gappto\expandafter\mymatrixcontent\expandafter{\mymatrixrow\\}} %

\matrix (m3)[matrix anchor= north west, matrix of math nodes,left delimiter=\lbrack, right delimiter=\rbrack, nodes in empty cells, row sep=0.2em, minimum size=1.3em, ampersand replacement=\&]
{ \mymatrixcontent};
{[loosely dotted, shorten >=15pt, shorten <=15pt, thick]
\draw (m3-2-1.north west) -- (m3-4-5.north west);
\draw (m3-5-7.north west) -- (m3-7-11.north west);
\draw (m3-8-13.north west) -- (m3-10-17.north west);
\draw (m3-10-18.south east) -- (m3-12-22.south east);
}
\begin{pgfonlayer}{myback}
	\draw[shift={(m3-5-6.north west)}, fill=black, bar width=10pt, xscale=0.2, yscale=0.3,  variable=\t, samples at={-5,...,5}, xshift=3pt, yshift=5pt ] plot[ybar] (\t, { 4*pow(abs(\t)+1.2, -2} ); %
	\draw[shift={(m3-8-12.north west)}, fill=black, bar width=10pt, xscale=0.2, yscale=0.3,  variable=\t, samples at={-5,...,5}, xshift=3pt, yshift=5pt ] plot[ybar] (\t, { 4*pow(abs(\t)+1.2, -2} ); %
	\draw[shift={(m3-11-18.north west)}, fill=black, bar width=10pt, xscale=0.2, yscale=0.3,  variable=\t, samples at={-5,...,5}, xshift=-3pt] plot[ybar] (\t, { 4*pow(abs(\t)+1.2, -2} ); %
\end{pgfonlayer}
\end{tikzpicture}
%
%
\caption{A $1/\lambda-$slanted matrix $M$. The matrix is dominated by entries on a slanted diagonal of slope $1/\lambda$. } 
\label{fig:skewmatrix}
\end{figure}
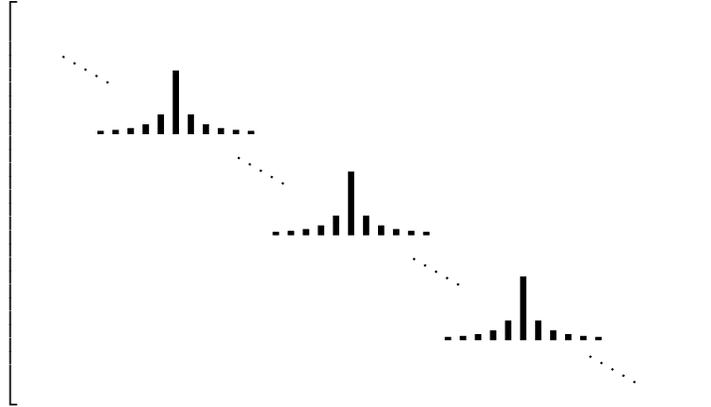

\begin{lemma}\label{lem:instability}
Given $M=(m_{j',j})_{j',j\in\Z^2}$.  If there exists a monotonically decreasing function
$w\colon R^+_0\longrightarrow R^+_0$ with $w=O(x^{-2-\delta})$, $\delta>0$, and constants
$\lambda>1$ and $K_0>0$ with $\abs{m_{j',j}}<w(\norm{\lambda j'-j}_\infty)$ for
$\norm{\lambda j'-j}_\infty>K_0$, then $M$ is not stable.
\end{lemma}
Intuitively, this result asserts that a bi-infinite matrix whose entries decay rapidly away from a skew diagonal
behaves like a finite matrix with more rows than columns (see Figure~\ref{fig:skewmatrix}).
Such a matrix will always have a nontrivial nullspace.  In the case of an infinite matrix what can be shown is
that at best its inverse will be unbounded.

We can make a more direct connection from this proof to the original necessity argument in \cite{Kai59} in the following way.
If we restrict our attention to sequences $\{\sigma_{k,l}\}$ with a fixed finite support of size say $N$, then the
image of this subspace of sequence space under the mapping $E$ is an $N$-dimensional subspace of $OPW(S)$.
The operator $P$ is essentially a time-frequency localization operator.  This fact is established in \cite{KP06}
and follows from the rapid decay of the Fourier transform of $\eta_P$.  Since $\eta_P$ itself is concentrated
on a rectangle of area $BL/\lambda^2$, its
Fourier transform will be concentrated on a rectangle of area $\lambda^2/BL$.  From this it follows that
for $\sigma$ as described above, the operator $E(\sigma)$ essentially localizes a function to a region in the
time-frequency plane of area $N(\lambda^2/BL)$.

Considering now the Gabor analysis operator $C_g$, we observe that the Gaussian $g(x)$ essentially occupies a time-frequency
cell of area $1$, and that this function is shifted in the time-frequency plane by integer multiples of $(\lambda^2/B, \lambda^2/L)$.
Hence to ``cover'' a region in the time-frequency plane of area $N(\lambda^2/BL)$ would require only about
$$\frac{N(\lambda^2/BL)}{\lambda^4/BL} = \frac{N}{\lambda^2}$$
time-frequency shifts.  So roughly speaking, in order to resolve $N$ degrees of freedom in the operator $E({\sigma_{k,l}})$, we have only
$N/\lambda^2 < N$ degrees of freedom in the output of the operator $E({\sigma_{k,l}})s$.

\subsection{Identification of operator Paley-Wiener spaces by periodically weighted delta-trains}\label{section:mainidentification}

Theorem~\ref{thm:main-simple} is based on arguments outlined in Section~\ref{section:KailathSufficient} and applies only to
$OPW(S)$ if $S$ is contained in a rectangle of area less than or equal to one.
In the following, we will develop the tools that allow us to identify $OPW(S)$ for any compact set $S$ of Lebesgue measure less than one.

In our approach we discretize the channel by covering the spreading support $S$ with small rectangles of fixed sidelength, which we refer to as a
{\em rectification} of $S$.  As long as the measure of $S$ is less than one, it is possible to do this in such a way that the total area of the rectangles
is also less than one.  This idea seems to bear some similarity to Bello's philosophy of sampling the spreading function on a fixed grid but with one fundamental difference.  Bello's approach is based on replacing $t$
and $x$ by samples, thereby approximating the channel. For a better approximation, sampling on a finer grid
is necessary, which results in a larger system of equations that must be solved.
In our approach, as soon as the total area of the rectification is less than one, the operator modeling the channel is completely determined
by the discrete model.  Once this is achieved,
identification of the channel reduces to solving a single linear system of equations at each point.

\begin{figure}
 \begin{center}
\includegraphics[height=4cm]{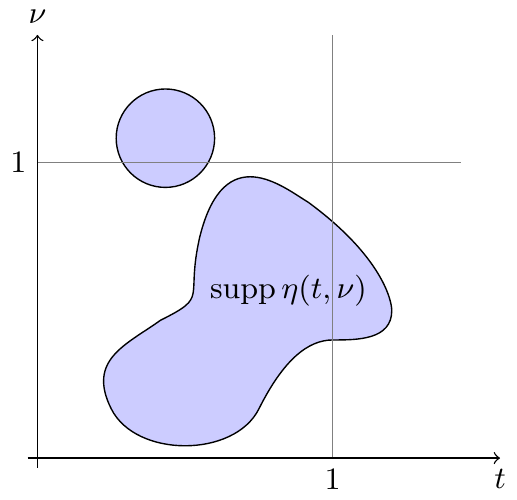}\qquad  \includegraphics[height=4cm]{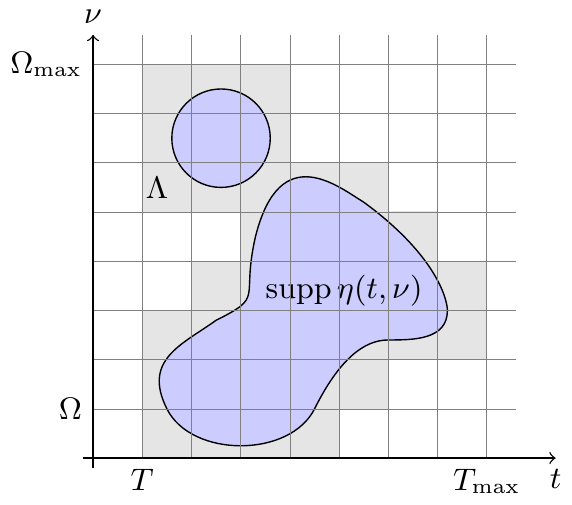}
\end{center}
\caption{A set not satisfying Kailath's condition is rectified with $ \ 1/(T\Omega)=P\in\N$,  the rectification has area $\leq 1$,  $\Omega_{\rm max}\leq1/T$, and  $T_{\rm max}\leq1/\Omega$.}
\end{figure}

Given parameters $T>0$ and $P\in\N$, we assume that $S$ is rectified by rectangles of size $T\times\Omega$, where $\Omega = 1/(TP)$, such that the
total area of the rectangles is less than one.  Given a period-$P$ sequence $c=(c_n)_{n\in\Z}$, we then define the {\em periodically weighted delta-train} $g$
by $g = \sum_{n\in\Z} c_n\,\delta_{nT}$.  The goal of this subsection is to describe the scheme by which a linear system of $P$ equations in a priori $P^2$ unknowns
can be derived by which an operator $H\in OPW(S)$ can be completely determined by $Hg(x)$.  In this sense, the ``degrees of freedom'' in the operator class
$OPW(S)$, and that of the output function $Hg(x)$ are precisely defined and can be effectively compared.

The basic tool of time-frequency analysis that makes this possible is the {\em Zak transform} (see \cite{Gro01}).

\begin{definition}\label{def:zaktransform}
The non-normalized Zak Transform is defined for $f\in{\cal S}(\R)$\footnote{${\cal S}(\R)$ denotes the Schwartz class of infinitely-differentiable,
rapidly-decreasing functions.}, and $a>0$ by
$$\displaystyle{Z_a f(t,\nu) = \sum_{n\in\Z} f(t-an)\,e^{2\pi ia n\nu}}.$$
\end{definition}

$Z_af(t,\nu)$ satisfies the quasi-periodicity relations
$$\displaystyle{Z_af(t+a,\nu) = e^{2\pi ia\nu}\,Z_af(t,\nu)}$$
and
$$\displaystyle{Z_af(t,\nu+1/a) = Z_af(t,\nu)}.$$  $\sqrt{a}\,Z_a$ can be extended to a unitary operator
from $L^2(\R)$ onto $L^2([0,a]{\times}[0,1/a])$.

A somewhat involved but elementary calculation yields the following (see \cite{PW14} and Section~\ref{section:proof of lemma}).

\begin{lemma}\label{lem:matrixequationquasiperiodic}
Let $T>0$, $P\in\N$, $c=(c_n)$, and $g$ be given as above.
Then for all $(t,\nu)\in\R^2$, and $p=0,\,1,\,\dots,\,P{-}1$,
\begin{align}\label{eqn:matrixequationquasiperiodic}
&       e^{-2\pi i\nu T p}\,(Z_{TP}\circ H)g(t + T p,\nu) \nonumber\\
&  =    \Omega\,\sum_{q,\,m=0}^{P-1} (T^q\,M^m c)_p\,
  e^{-2\pi i\nu T q}\,\eta^{Q}_H(t + T q,\nu + m/TP).
  \end{align}
\end{lemma}
Here $\mathcal T$ and $\mathcal M$ are the translation and modulation operators given in Definition~\ref{def:translationmodulation}, and $\eta^{Q}_H(t,\nu)$
is the {\em quasiperiodization} of $\eta_H$,
\begin{equation}\label{eqn:quasiperiodization}
\eta^{Q}_H(t,\nu) = \sum_k\sum_\ell \eta_H(t+kTP,\nu+\ell/T)\,e^{-2\pi i\nu kTP}
\end{equation}
whenever the sum is defined.

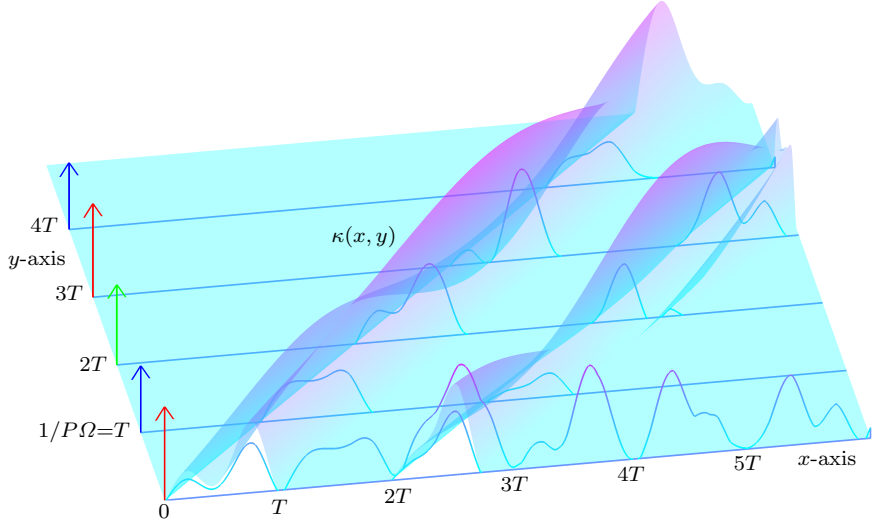
\begin{figure}\label{fig:PWSounding}
\begin{tikzpicture}
\node[above right]  at (0,1){
\includegraphics[width=11cm]{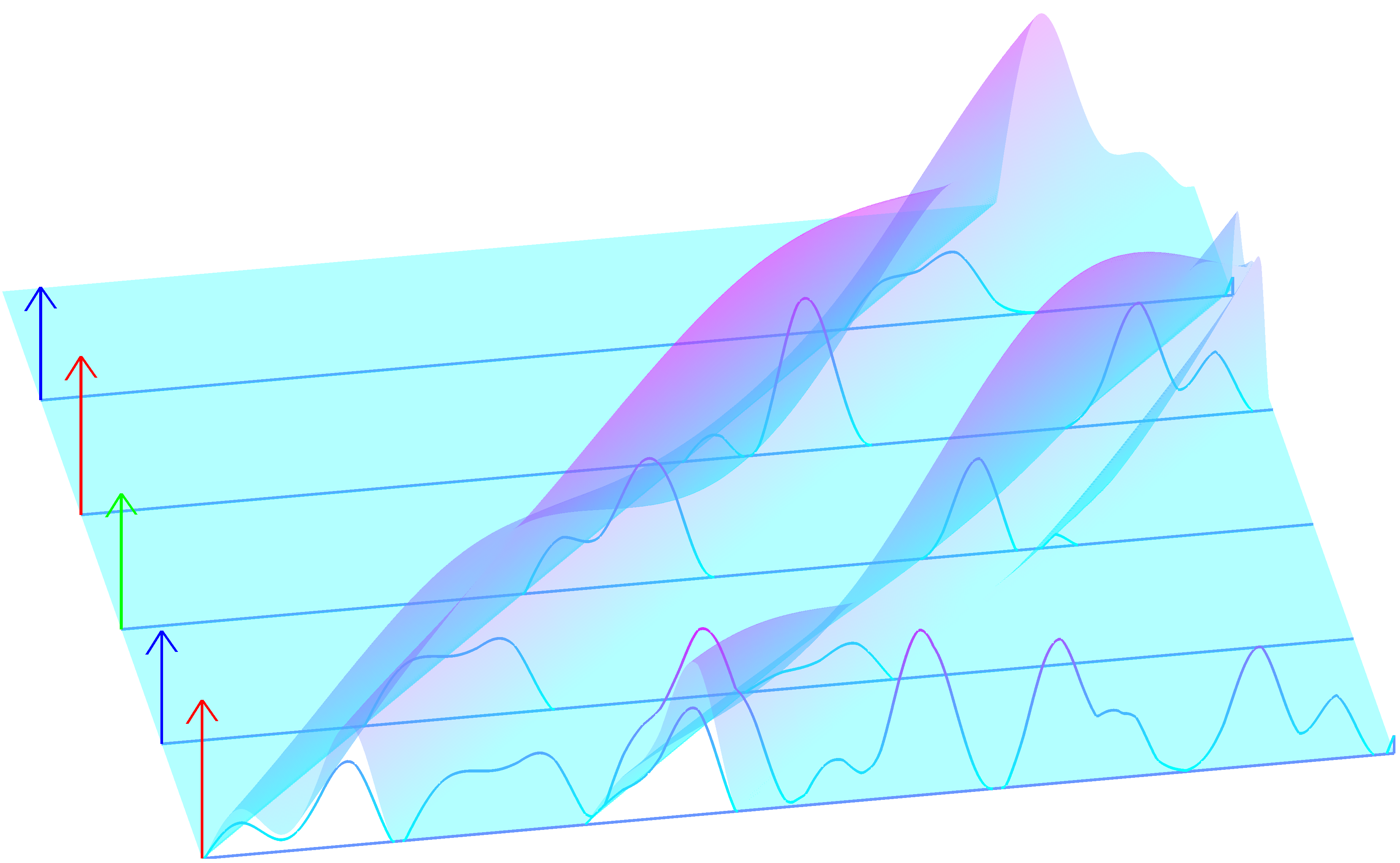}
};
\node[above right]  at (10,1.5) {$x$-axis};
\node[above right] at (-.5,4.1){$y$-axis};
\node[above right]  at (3.8,4.4){$\kappa(x,y)$};
\node[above right]  at (1.5,.8){$0$};
\node[above right]  at (-0.1,1.8){$1/P\Omega{=}T$};
\node[above right]  at (.45,2.75){$2T$};
\node[above right]  at (0.13,3.7){$3T$};
\node[above right]  at (-.2,4.6){$4T$};
\node[above right]  at (3,.85){$T$};
\node[above right]  at (4.5,1){$2T$};
\node[above right]  at (6.05,1.15){$3T$};
\node[above right]  at (7.6,1.28){$4T$};
\node[above right]  at (9.15,1.43){$5T$};
\end{tikzpicture}

\caption{Channel sounding  of $OPW([0,2/3]{\times}[-1/4,1/4]\, \cup\, [4/3,2]{\times}[-1/2,1/2])$ using a $P$-periodically weighted delta train $g$.    The kernel $\kappa(x,y)$ takes values on the $(x,y)$-plane, the sounding signal $g$, a weighted impulse train, is defined on the $y$-axis, and the output signal $Hg(x)=\int \kappa(x,y)g(y)dy$ is displayed on the $x$-axis. Here, the sample values of the tab functions $h(x,t)=\kappa(x,t-x)$ are not easily read of the response $Hg(x)$ as, for example for $x\in [2T,3T]=[4/3,2]$ we have $Hg(x)=0.7\kappa(x,0)+0.6\kappa(x,2T)=0.7h(x,x)+.6h(x,2T-x)$.  In detail, we have 
 $g = \ldots +0.7\delta_{-2}+0.5\delta_{-4/3}+ 0.6\delta_{-2/3} +0.7\delta_0+0.5\delta_{2/3}+ 0.6\delta_{4/3}+0.7\delta_{2}+ 0.5\delta_{8/3}+\ldots$, so $P=3$, $T=2/3$, $\Omega=1/PT=1/2$, $c_n=0.7$ if $n \!\! \mod 3=0$, $c_n=0.5$ if $n\!\!\mod 3=1$, $c_n=0.6$ if $n\!\! \mod 3=2$.
.}
\end{figure}

Under the additional simplifying assumption that the spreading function $\eta_H(t,\nu)$ is supported in the large rectangle
$[0,TP]\times [0,1/T]$, and by restricting (\ref{eqn:matrixequationquasiperiodic}) to the rectangle $[0,T]\times[0,1/(TP)]$,
we arrive at the $P\times P^2$ linear system
\begin{equation}\label{eqn:basiclinearsystem}
{\bf Z}_{Hg}(t,\nu)_p = \sum_{q,m=0}^{P-1} G(c)_{p,(q,m)}\,\boldsymbol\eta_H(t,\nu)_{(q,m)}
\end{equation}
where
\begin{equation}\label{eqn:boldZvector}
{\bf Z}_{Hg}(t,\nu)_p = (Z_{TP}\circ H)g(t+pT,\nu)\,e^{-2\pi i \nu pT},
\end{equation}
\begin{equation}\label{eqn:boldetavector}
\boldsymbol\eta_H(t,\nu)_{(q,m)} = \Omega\,\eta_H(t+qT,\nu+m/TP)\,e^{-2\pi i \nu qT}\,e^{-2\pi iqm/P},
\end{equation}
and where $G(c)$ is a finite Gabor system matrix (\ref{def:fullgabormatrix}).
If (\ref{eqn:basiclinearsystem}) can be solved for each $(t,\nu)\in [0,T]\times[0,1/(TP)]$,
then the spreading function for an operator $H$ can be completely determined by its response to the
periodically-weighted delta-train $g$.

As anticipated by Bello, two issues now become relevant.  (1) We require that $\supp\eta_H$ occupy no more than $P$ of the shifted rectangles
$[0,T]\times[0,1/(TP)]+(qT,k/(TP))$, so that (\ref{eqn:basiclinearsystem}) has at least as many equations as unknowns.  This forces $|\supp\eta_H|\le 1$. (2) We require that $c$ be chosen in such a way that the $P\times P$ system formed by removing the columns of
$G(c)$ corresponding to vanishing components of $\boldsymbol\eta_H$ is invertible.  That such $c$ exist is a fundamental cornerstone of
operator sampling and is the subject of the next section.

Based on the existence of $c$ such that any set of $P$ columns of $G(c)$ form a linearly independent set, we can prove the following \cite{PW13}.

\begin{theorem}\label{thm:reconstruction}
For $S\subseteq (0,\infty){\times}\R$  compact with $\abs{S}<1$, there exists $T>0$ and $P\in\N$ 
, and a period-$P$ sequence
$c=(c_n)$ such that $g=\sum_n c_n\,\delta_{nT}$ identifies $OPW(S)$.
In particular, there exist period-$P$ sequences $b_j=(b_{j,k})$, and integers $0\le q_j,\,m_j\le P{-}1$, for $0\le j\le P{-}1$ such that
\begin{align}
h(x,t)  &=  e^{-\pi it/T}\sum_k\sum_{j=0}^{P-1} \big[b_{j,k}\,Hg(t - (q_j-k)T) \nonumber  \\
        & \hskip-.25in e^{2\pi i m_j(x-t)/PT}\,\phi((x-t)+(q_j-k)T)\, r(t-q_j T)\big] \label{eqn:reconstructionformula}
\end{align}
where  $r,\phi\in {\cal S}(\R)$ satisfy
\begin{equation}\label{eq:r_phi_2}
\sum_{k\in\Z} r(t + kT) = 1 = \sum_{n\in\Z} \widehat{\phi}(\gamma + n/PT),
\end{equation}
where $r(t)\widehat{\phi}(\gamma)$ is supported in a neighborhood of $[0,T]{\times}[0,1/PT]$, and where the
sum in \eqref{eqn:reconstructionformula} converges unconditionally in $L^{2}$ and for each $t$ uniformly in $x$.
\end{theorem}

Equation~\eqref{eqn:reconstructionformula} is a generalization of \eqref{eqn:operatorreconstruction-simple} which is easily seen by choosing $\phi(x)=\sin(\pi PTx)/(\pi PTx)$ and $r(t)$ to be the characteristic function of $[0,T)$.

\section{Linear Independence Properties of Gabor Frames}\label{section:finiteGabor}

\subsection{Finite Gabor Frames}

\begin{definition}\label{def:translationmodulation}
Given $P\in\N$, let $\omega=e^{2\pi i/P}$ and define the {\em translation operator $\mathcal T$}
on $(x_0,\,\dots,\,x_{P-1})\in\mathbb C^P$ by
$$\mathcal{T}x=(x_{P-1},x_0,\,x_1,\,\ldots,x_{P-2}),$$
and the {\em modulation operator $\mathcal M$} on $\mathbb C^P$ by
$$\mathcal{M}x=(\omega^0 x_0, \omega^1 x_1,\,\dots,\, \omega^{P-1} x_{P-1}).$$
Given a vector $c\in\mathbb C^P$ the {\em finite Gabor system with window $c$}
is the collection $\set{\mathcal{T}^q \mathcal{M}^p c}_{q,p=0}^{P-1}$.
Define the {\em full Gabor system matrix} $G(c)$
to be the $P\times P^2$ matrix
\begin{equation}\label{def:fullgabormatrix}
G(c) = \left[\,\, D_0\,W_P\,\,\vrule\,\, D_1\,W_P\,\, \vrule \,\,\cdots \,\, \vrule\,\, D_{P-1}\,W_P \,\,\right]
\end{equation}
where $D_k$ is the diagonal matrix with diagonal
$$\mathcal{T}^kc = (c_{P-k},\,\dots,\,c_{P-1},\,c_0,\,\dots,\,c_{P-k-1}),$$
and $W_P$ is the $P\times P$ Fourier matrix $W_P = (e^{2\pi inm/P})_{n,m=0}^{P-1}$.
\end{definition}

\begin{remark}\label{rem:translationmodulation}
(1) For $0\le q,\,p\le P-1$, the $(q+1)$st column of the submatrix $D_pW_P$ is the vector $\mathcal{M}^p\mathcal{T}^qc$
where the operators $\mathcal{M}$ and $\mathcal{T}$
are as in Definition~\ref{def:translationmodulation}.  This means that each column of the matrix $G(c)$
is a unimodular constant multiple of an element of the finite Gabor system with window $c$, namely
$\set{e^{-2\pi ipq/P}\,\mathcal{T}^q \mathcal{M}^pc}_{q,p=0}^{P-1}$.

\noindent (2)
Note that the finite Gabor system defined above consists of $P^2$ vectors in $\mathbb C^P$ which form an overcomplete
tight frame for $\mathbb C^P$ \cite{LPW05}.  For details on Gabor frames in finite dimensions, see \cite{LPW05,KPR08,FKL09}
and the overview article \cite{Pfa12}.

\noindent(3)  Note that we are abusing notation slightly by identifying a vector $c\in\mathbb C^P$ with an $P$-periodic sequence
$c=(c_n)$ in the obvious way.
\end{remark}

\begin{definition}\cite{DE03}
The {\em Spark} of an $M\times N$ matrix F is the size of the smallest linearly dependent
subset of columns, i.e.,
$$Spark(F) = \min\set{\norm{x}_0\colon Fx=0,\ \ x\ne 0}$$
where $\norm{x}_0$ is the number of nonzero components of the vector $x$.
If $Spark(F)=M+1$, then $F$ is said to have {\em full Spark}.
$Spark(F)=k$ implies that any collection of fewer than $k$ columns of $F$
is linearly independent.
\end{definition}

\subsection{Finite Gabor frames are generically full Spark}

The existence of Gabor matrices with full Spark has been addressed in \cite{LPW05,M13}.  The results in these
two papers are as follows.

\begin{theorem}\label{thm:genericfullsparkprime}\cite{LPW05}
If $P\in\N$ is prime then there exists a dense, open subset of $c\in\mathbb C^P$ such that every minor of the Gabor system matrix $G(c)$
is nonzero.  In particular, for such $c$, $G(c)$ has full Spark.\end{theorem}
\begin{theorem}\label{thm:genericfullspark}\cite{M13}
For every $P\in\N$, there exists a dense, open subset of $c\in\mathbb C^P$ such that the Gabor system matrix $G(c)$ has full Spark.
\end{theorem}

The goal of this subsection is to outline the proof of Theorems~\ref{thm:genericfullsparkprime} and \ref{thm:genericfullspark}.
We will adopt some of the following notation and terminology of \cite{M13}.

Let $P\in\N$ and let $M$ be an $P\times P$ submatrix of $G(c)$.  For $0\le \kappa<P$ let $\ell_\kappa$ be the number of columns of $M$
chosen from the submatrix $D_\kappa W_P$ of (\ref{def:fullgabormatrix}).  While the vector $\ell=(\ell_\kappa)_{\kappa=0}^{P-1}$
does not determine $M$ uniquely, it describes the matrix $M$ sufficiently well for our purposes.
Define $M_\kappa$ to be the $P\times\ell_\kappa$ matrix consisting of those columns of $M$ chosen from $D_\kappa W_P$.
Given the {\em ordered partition} $B=(B_0,\,B_1,\,\dots,\,B_{P-1})$
where $\set{B_0,\,B_1,\,\dots,\,B_{P-1}}$
forms a partition of $\set{0,\,\dots,\,P-1}$, and where for each $0\le\kappa<P$, $|B_\kappa|=\ell_\kappa$, let $M_\kappa(B_\kappa)$
be the $\ell_\kappa\times\ell_\kappa$ submatrix of $M_\kappa$ whose rows belong to $B_\kappa$.  Then
$\det(M) = \prod \det(M_\kappa(B_\kappa))$ where the product is taken over all such ordered partitions $B$.
This formula is called the {\em Lagrange expansion} of the determinant.

Each ordered partition $B$ corresponds to a permutation on $\Z_P$ as follows.
Define the {\em trivial partition} $A=(A_0,\,A_1,\,\dots,\,A_{P-1})$ by
$$A_j = \set{\sum_{i=0}^{j-1}\ell_i, \big(\sum_{i=0}^{j-1}\ell_i\big)+1,\,\dots,\,\big(\sum_{i=0}^j\ell_i\big)-1}$$
so that $A_0=[0,\ell_0-1]$, $A_1=[\ell_0, \ell_0+\ell_1+1]$, $\dots$, $A_{P-1}=[\ell_0+\,\cdots\,+\ell_{P-2},P-1]$.  Then given
$B=(B_0,\,B_1,\,\dots,\,B_{P-1})$ there is a permutation $\sigma\in S_P$ such that $\sigma(A_\kappa)=B_\kappa$ for all $\kappa$.  This $\sigma$
is unique up to permutations that preserve $A$, that is, up to $\tau\in S_P$ such that $\tau(A_\kappa)=A_\kappa$ for all $\kappa$.  Call such
a permutation {\em trivial} and denote by $\Gamma$ the subgroup of $S_P$ consisting of all trivial permutations.  Then the ordered partitions
$B$ of $\Z_P$ can be indexed by equivalence classes of permutations $\sigma\in S_P/\Gamma$.

The key observation is that $\det(M)$ is a homogeneous polynomial in the $P$ variables $c_0,\,c_1,\,\dots,\,c_{P-1}$ and we can write
\begin{equation}\label{eqn:determinant}
\det(M) = \sum_{\sigma\in S_P/\Gamma} a_\sigma\,C^\sigma
\end{equation}
where the monomial $C^\sigma$ is given by
$$C^\sigma = \prod_{\kappa=0}^{P-1}\,\prod_{j\in \sigma(A_\kappa)} c_{(j-\kappa)(mod\ P)}.$$
If it can be shown that this polynomial does not vanish identically then we can choose a dense, open subset of $c\in\mathbb C^P$ for which $\det(M)\ne 0$.
Since there are only finitely many $P\times P$ submatrices of $G(c)$ it follows that there is a dense, open subset of $c$ for which $\det(M)\ne 0$
for all $M$, and we conclude that, for these $c$, $G(c)$ has full Spark.

Following \cite{M13}, we say that a monomial $C^{\sigma_0}$ {\em appears uniquely} in (\ref{eqn:determinant}) if for every
$\sigma\in S_P/\Gamma$ such that $\sigma\ne\sigma_0$, $C^{\sigma}\ne C^{\sigma_0}$.
Therefore, in order to show that the polynomial (\ref{eqn:determinant}) does not vanish identically, it is sufficient to show that
(1) there is a monomial $C^\sigma$ that appears uniquely in (\ref{eqn:determinant}) and (2) the coefficient $a_\sigma$ of this monomial does not vanish.

Obviously, whether or not (\ref{eqn:determinant}) vanishes identically does not depend on how the variables $c_i$ are labelled.  More specifically,
if the variables are renamed by a cyclical shift of the indices, viz., $c_i \mapsto c_{(i+\gamma)mod\ P}$ for some $0\le\gamma<P$, then
$$\det(M)(c_{\gamma+1},\,\dots,\,c_{P-1},\,c_0,\,\dots,\,c_\gamma) = \pm\,\det(M')(c_0,\,\dots,\,c_{P-1})$$
where $M'$ is an $P\times P$ submatrix described by the vector
$$\ell'=(\ell_{\gamma+1},\,\dots,\,\ell_{P-1},\,\ell_0,\,\dots,\,\ell_\gamma).$$

\subsubsection{The lowest index monomial}

In \cite{LPW05}, a monomial referred to in \cite{M13} as the {\em lowest index (LI) monomial} is defined that has the required properties when
$P$ is prime.  In order to see this, note first that each coefficient $a_\sigma$ in  the sum (\ref{eqn:determinant}) is the product of minors
of the Fourier matrix $W_P$ and since $P$ is prime, Chebotarev's Theorem says that such minors do not vanish \cite{SL96}.  More specifically,
$$a_\sigma\,C^\sigma = \pm\,\prod_{\kappa=0}^{P-1} \det(M_\kappa(\sigma(A_\kappa)))$$
and for each $\kappa$, the columns of $M_\kappa$ are columns of $W_P$ where each row has been multiplied by the same
variable $c_j$ and $M_\kappa(\sigma(A_\kappa))$
is a square matrix formed by choosing $\ell_\kappa$ rows of $M_\kappa$.  Hence for each $\kappa$, $\det(M_\kappa(\sigma(A_\kappa)))$ is a
monomial in $c$ with coefficients a constant multiple of a minor of $W_P$.  Since $a_\sigma$ is the product of those minors, it does not vanish.

Note moreover that each submatrix $M_\kappa(\sigma(A_\kappa))$ is an $\ell_\kappa\times\ell_\kappa$ matrix, so that
$\det(M_\kappa(\sigma(A_\kappa)))$ is the sum of a multiple of the product of $\ell_\kappa!$ diagonals of $M_\kappa(\sigma(A_\kappa))$.  Hence
$a_\sigma\,C^\sigma$ is the sum of multiples of the product of $\prod_{\kappa=0}^{P-1} \ell_\kappa!$ generalized diagonals of $M$.

We define the LI monomial as in \cite{LPW05} as follows.  If $M$ is $1\times 1$, then $\det(M)$ is a multiple of a single variable $c_j$
and we define the LI monomial, $p_M$ by $p_M=c_j$. If $M$ is $d\times d$, let $c_j$ be the variable of lowest index appearing in $M$.
Choose any entry of $M$ in which $c_j$ appears, eliminate the row and column containing that entry,
and call the remaining $(d-1)\times(d-1)$ matrix $M'$.  Define $p_M = c_j\,p_{M'}$.  It is easy to see that the monomial $p_M$ is independent
of the entry of $M$ chosen at each step.  In order to show that the LI monomial appears uniquely in (\ref{eqn:determinant}), we observe as in
\cite{LPW05} that the number of diagonals in $\det(M)$ that correspond to the LI monomial is $\prod_{\kappa=0}^{P-1}\ell_\kappa!$.
Since this is also the number of generalized diagonals appearing in the calculation of each $\det(M_\kappa(\sigma(A_\kappa)))$, it follows
that this monomial appears only once.  The details of the argument can be found in Section~\ref{section:prime}.
Note that because $P$ is prime, this argument is valid no matter how large the matrix $M$ is.  In other words, $M$ does not have
to be an $P\times P$ submatrix in order for the result to hold.  Consequently, given $k<P$ and $M$ an arbitrary $P\times k$ submatrix of $G(c)$,
we can form the $k\times k$ matrix $M_0$ by choosing $k$ rows of $M$ in such a way that the LI monomial of $M_0$ contains at most only the
variables $c_0,\,\dots,\,c_{k-1}$.  This observation leads to the following theorem for matrices with arbitrary Spark.

\begin{theorem}\label{thm:smallsparksampta13}\cite{PW14}
If $P\in\N$ is prime, and $0<k< P$, there exists an open, dense subset of $c\in\C^k\times\set{0}\subseteq \mathbb C^P$
with the property that $Spark(G(c))=k+1$.
\end{theorem}

This result has implications for relating the capacity of a time-variant communication channel to the area of the spreading support, see
\cite{PW14}.

\subsubsection{The consecutive index monomial}

In \cite{M13}, a monomial referred to as the {\em consecutive index (CI) monomial} is defined that has the required properties for any $P\in\N$.
The CI monomial, $C^I$, is defined as the monomial corresponding to the identity permutation in $S_P/\Gamma$, that is,
to the equivalence class of the trivial partition $A=(A_0,\,A_1,\,\dots,\,A_{P-1})$.  Hence
$$C^I = \prod_{\kappa=0}^{P-1}\,\prod_{j\in A_\kappa} c_{(j-\kappa)mod\ P}.$$
For each $\kappa$, the monomial appearing in
$\det(M_\kappa(A_\kappa))$, $\prod_{j\in A_\kappa} c_{(j-\kappa)mod\ P}$, consists of a product of $\ell_k$ variables $c_j$ with consecutive indices
modulo $P$.

That $a_I\ne 0$ follows from the observation that for each $\kappa$, $\det(M_\kappa(A_\kappa))$ is a monomial whose coefficient
is a nonzero multiple of a Vandermonde determinant and hence does not vanish (for details, see \cite{M13}).
The proof that $C^I$ appears uniquely in (\ref{eqn:determinant})
amounts to showing that, with respect to an appropriate cyclical renaming of the variables $c_i$, the $CI$ monomial uniquely minimizes
the quantity $\Lambda(C^\sigma)=\sum_{i=0}^{P-1} i^2\,\alpha_i$, where $\alpha_i$ is the exponent of the variable $c_i$ in $C^\sigma$.
An abbreviated version of the proof of this result as it appears in \cite{M13} is given in Section~\ref{sec:malikiosis}.

As a final observation, we quote the following corollary that provides an explicit construction of a unimodular vector $c$
such that $G(c)$ has full Spark.

\begin{corollary}\cite{M13}
Let $\zeta=e^{2\pi i/(P-1)^4}$ or any other primitive root of unity of order $(P-1)^4$ where $P\ge 4$.  Then the vector
$$c=(1,\,\zeta,\,\zeta^4,\,\zeta^9,\,\dots,\,\zeta^{(P-1)^2})$$
generates a Gabor frame for which $G(c)$ has full Spark.
\end{corollary}

\section{Generalizations of operator sampling to higher dimensions}\label{section:higherdimensional}

The operator representations \eqref{eqn:operator1}, \eqref{eqn:operator2}, and \eqref{eqn:operator3} hold verbatim for higher dimensional variables $x,\xi,t,\nu\in \R^d$. In this section, we address the identifiability of
 \begin{align*}
OPW (S) & =  \{H\in \mathcal L (L^2(\R^d), L^2(\R^d)) \colon
         \supp\mathcal F_s \sigma_H\subseteq S,\,\norm{\sigma_H}_{L^2}<\infty\}
\end{align*}
where $S\subseteq \R^{2d}$.

Looking at the components of the multidimensional variables separately, Theorem~\ref{thm:main-simple} easily generalizes as follows.

\begin{theorem}\label{thm:main-simple-higher}
  For $H\in OPW \big(\prod_{\ell=1}^d [0, T_\ell] {\times}\prod_{\ell=1}^d[-  \Omega_\ell / 2, \Omega_\ell / 2]\big)$ with  $T_\ell\Omega_\ell{\leq} 1$, $\ell=1,\ldots,d$, we have
  \begin{align*}
    \|H\sum_{k_1\in\Z}\ldots \sum_{k_d\in\Z}\delta_{(k_1T_1,\ldots, k_dT_d)}\|_{L^2(\R)}=T_1\ldots T_d \|\sigma_H\|_{L^2},\notag 
  \end{align*} and  $H$ can be reconstructed by means of
  \begin{align*}
    \kappa_H(x+t,x)&=\chi_{\prod_{\ell=1}^d [0, T_\ell]}(t)\sum_{n_1\in\Z}\ldots \sum_{n_d\in\Z} \\ \quad &\big(H \sum_{k_1\in\Z}\ldots \sum_{k_d\in\Z}\delta_{(k_1T_1,\ldots, k_dT_d)} \big)(t+(n_1T_1,\ldots,n_d T_d) \\ &\quad
     \frac{\sin(\pi T_1 (x_1-n_1))}{\pi T_1 (x_1-n_1)} \ldots \frac{\sin(\pi T_d (x_d-n_d))}{\pi T_d (x_d-n_d)} \notag
  \end{align*}
  with convergence in the $L^2$  norm.
\end{theorem}

In the following, we address the situation where $S$ is not contained in a set  $\prod_{\ell=1}^d [0, T_\ell] {\times}\prod_{\ell=1}^d[-  \Omega_\ell / 2, \Omega_\ell / 2]\big)$ with  $T_\ell\Omega_\ell{\leq} 1$, $\ell=1,\ldots, d$. For example, $S=[0,1]\times [0,2]\times[0,\frac 1 4]\times[0,1]\subseteq \R^4$ of volume $\frac 1 2$ is not covered by Theorem~\ref{thm:main-simple-higher}.

To give a higher dimensional variant of Theorem \ref{thm:reconstruction}, we shall denote pointwise products of finite and infinite length vectors $k$ and $T$ by $k\pw T$, that is, $k\pw T=(k_1T_1,\ldots,k_d T_d)$ for $k,T\in \mathbb C^d$. Similarly, $k/ T=(k_1/T_1,\ldots,k_d /T_d)$.

\begin{theorem}\label{thm:reconstruction-higher}
If $S\subseteq (0,\infty)^d {\times}\R^d$ is compact with $\abs{S}<1$ then $OPW (S)$ is identifiable.
Specifically, there exist $T_1,\ldots,T_d>0$ and pairwise relatively prime natural numbers $P_1,\ldots, P_d$ such that
$$S\subseteq \prod_{\ell=1}^d[0,P_\ell T_\ell]{\times}\prod_{\ell=1}^d [-1/(2T_\ell), 1/(2T_\ell)],$$
and a sequence $c=(c_n) \in \ell^\infty(\Z^d)$ which is $P_\ell$ periodic in the $\ell$-th component $n_\ell$
such that $g=\sum_{n\in\Z^d} c_n\,\delta_{n \pw T}$ identifies $OPW^{2}(S)$.
In fact, for such $g$ there exists for each $j\in J=\prod_{\ell =1}^d \{0,1,\ldots,P_\ell{-}1\}$ a sequences
$b_j=(b_{j,k})$ which is $P_\ell$ periodic in $k_\ell$ and $2d$-tuples $(q_j,m_j) \in J\times J$ with
\begin{align}
h(x,t)  &=  e^{-\pi i \sum_{\ell=1}^d t_\ell/T_\ell}\sum_{k\in\Z^d}\sum_{j\in J} \big[b_{j,k}\,Hg(t - (q_j-k)\pw T) \nonumber  \\
        & \hskip-.25in e^{2\pi i m_j\cdot ( (x-t)/P\pw T)}\,\phi((x-t)+(q_j-k)\pw T)\, r(t-q_j \pw T)\big]. \label{eqn:reconstruction-higher}
\end{align}
The functions
  $r,\phi\in {\cal S}(\R^d)$ are assumed to satisfy
\begin{equation}\label{eq:r_phi_2}
\sum_{k\in\Z^d} r(t + k \pw T) = 1 = \sum_{n\in\Z^d} \widehat{\phi}(\gamma + (n/P\pw T),
\end{equation}
 and $r(t)\widehat{\phi}(\gamma)$ is supported in a neighborhood of $\prod_{\ell=1}^d [0,T_\ell]{\times}\prod_{\ell=1}^d[0,1/P_\ell T_\ell]$. The
sum in \eqref{eqn:reconstruction-higher} converges unconditionally in $L^{2}$ and for each $t$ uniformly in $x$.
\end{theorem}

This result follows from adjusting the proof of  Theorem \ref{thm:reconstruction-higher} to the higher dimensional setting.  For example, it will employ the Zak transform
$$\displaystyle{Z_{T \pw P} f(t,\nu) = \sum_{n\in\Z^d} f(t-n\pw P \pw T)\,e^{2\pi i  \nu \cdot (P\pw T)}},$$ where $P=(P_1,\ldots, P_d)$.
We are then led again to a system of linear equations of the form
\begin{equation}\label{eqn:basiclinearsystemhigher}
{\bf Z}_{Hg}(t,\nu)_p = \sum_{q\in J}  \sum_{m \in J}  G(c)_{p,(q,m)}\,\boldsymbol\eta_H(t,\nu)_{(q,m)}
\end{equation}
where as before
\begin{equation*}
{\bf Z}_{Hg}(t,\nu)_p = (Z_{T \pw P}\circ H)g(t+p\pw  T,\nu)\,e^{-2\pi i \nu p\pw  T},
\end{equation*}
\begin{align*}
\boldsymbol\eta_H(t,\nu)_{(q,m)} = &(T_1 P_1 \ldots T_d P_d )^{-1}\,\eta_H(t+q\pw T, \nu+(m/T\pw P )\,  \\ & e^{-2\pi i \nu \cdot( q\pw T)}\,e^{-2\pi i q \cdot (m/P)},
\end{align*}
and where $G(c)$ is now a multidimensional finite Gabor system matrix similar to (\ref{def:fullgabormatrix}).

In order to show that the spreading function for operator $H$ can be completely determined by its response to the
periodically-weighted $d$-dimensional delta-train $g$, we need to show that (\ref{eqn:basiclinearsystemhigher}) can be solved for each $(t,\nu)\in \prod_{\ell=1}^d [0,T_\ell]{\times}\prod_{\ell=1}^d[0,1/(T_\ell P_\ell)]$ if $c\in \mathbb C^{P_1\times \ldots \times P_d}$ is chosen appropriately.

To see that a choice of $c$ is possible, observe that the product group $\Z_{P_1}\times \ldots \times \Z_{P_d}$ is isomorphic to the cyclic group
$\Z_{P_1\cdot \ldots \cdot P_d}$ since the $P_\ell$ are chosen pairwise relatively prime.
Theorem~\ref{thm:genericfullspark} applied to the cyclic group  $\Z_{P_1\cdot \ldots \cdot P_d}$ guarantees  the existence of
$\widetilde c \in \mathbb C^{P_1\cdot \ldots \cdot P_d}$ so that the Gabor system matrix $G(\widetilde c)$ is full spark.
We can now define $c\in \mathbb C^{P_1\times \ldots \times P_d}$ by setting
$$c_{n_1,\ldots,n_d}=\widetilde{c}_{n_1+n_2\, P_1+n_3\,P_1P_2+\ldots + n_d\,P_1\ldots P_{d-1}}, \quad n=(n_1,\ldots,n_d)\in J $$ and observe that $G(c)$ is simply a rearrangement of $G(\widetilde c)$, hence, $G(c)$ is full spark.

\section{Further results on operator sampling}\label{section:outlook}

The results discussed in this paper are discussed in detail in \cite{Kai62,Bel69,KP06,PW06b,Pfa10} and \cite{PW14}.  The last listed article contains the most extensive collection of operator reconstruction formulas, including extensions to some $OPW(S)$ with $S$ unbounded. Moreover, some  hints on how to use parallelograms to rectify a set $S$ for operator sampling efficiently  are given.

A central result in \cite{PW14} is the classification of all spaces  $OPW(S)$ that are identifiable for a given $g=\sum_{n\in\Z} {c_n}\delta_{nT}$ for $c_n$ being $P$-periodic.

The papers \cite{PW06,Pfa10} address some functional analytic challenges in operator sampling, and \cite{KP12} focuses on the question of operator identification if we are restricted to using more realizable identifiers, for example, truncated and modified versions of $g$, namely,
$\widetilde {g} (t)=\sum_{n=0}^N {c_n} \varphi( t-nT)$.
The problem of recovering parametric classes of operators in $OPW(S)$ is discussed in \cite{BGE11,BGE11b}.

In the following, we briefly review literature that address some other directions in operator sampling.

\subsection{Multiple Input Multiple Output}

A Multiple Input Multiput Output (MIMO) channel $\bf H$ with $N$ transmitters and $M$ receivers can be modeled by an $N\times M$ matrix whose entries
are time-varying channel operators $H_{mn}\in OPW(S_{mn})$. For simplicity, we write ${\bf H} \in OPW(\bf S)$. Assuming that the operators $H_{mn}$
are independent, a sufficient criterion for identifiability is given by $\sum_{n=1}^N |S_{mn}| \leq 1$ for $m=1,\ldots, M$.
Conversely, if for a single $m$, $\sum_{n=1}^N |S_{mn}| > 1$, then $OPW(\bf S)$ is not identifiable by any collection $s_1,\ldots, s_N$ of input signals \cite{PW07,Pfa08}.

A somewhat dual setup was considered in  \cite{HP10}. Namely, a Single Input Single Output (SISO) channel with    $S$ being large, say $S=[0,M]\times [-N/2,N/2]$ with $N,M\geq 2$.  As illustrated above, $OPW([0,M]\times [-N/2,N/2])$ is not identifiable, but if we are allowed to use $MN$ (infinite duration) input signals $g_1,\ldots, g_{MN}$,  then $H\in OPW([0,M]\times [-N/2,N/2])$ can be recovered from the $MN$ outputs $Hg_1,\ldots, Hg_{MN}$.

\subsection{Irregular Sampling of Operators}

The identifier $g=\sum_{n\in\Z} c_n \delta_{nT}$ is supported on the lattice $T\Z$ in $\R$. In general, for stable operator identification, choosing a discretely supported identifier is reasonable, indeed, in \cite{KP12} it is shown that identification for $OPW(S)$ in full requires the use of  identifiers that neither decay in time nor in frequency. (Recovery of the characteristics of $H$ during a fixed transmission band and a fixed transmission interval can be indeed recovered when using Schwartz class identifiers \cite{KP12}.)

In irregular operator sampling, we consider identifiers of the form $g=\sum_{n\in\Z} c_n \delta_{\lambda_n}$ for nodes $\lambda_n$ that are not necessarily contained in a lattice. If such $g$ identifies $OPW(S)$, then we  refer to  $\supp g=\{\lambda_n\}$ as a sampling set for $OPW(S)$, and similarly, the {\em sampling rate} of $g$ is  defined to be
$$D(g)=D(\supp g)=D(\Lambda) = \lim_{r\to \infty} \frac{n^-(r)}{r}$$
where
$$n^-(r) = \inf_{x\in\R}\# \{ \Lambda \cap [x,x+r]\}$$
assuming that the limit exists  \cite{HP10,PW14}.

To illustrate a striking difference between irregular  sampling of functions and operators, note that $\Z$ is a sampling set for $OPW([0,1]\times [-\frac 1 2, \frac 1 2])$ as well as for the Paley Wiener space $PW([-\frac 1 2, \frac 1 2])$, but  the distribution $
g=c_0\delta_{\lambda_0}+\sum_{n\in\Z\setminus\{0\}} c_n \delta_{n}$ does not identify  $OPW([0,1]\times [-\frac 1 2, \frac 1 2])$, regardless of our choice of $c_n$ and $\lambda_0\neq 0$.
  This shows that, for example, Kadec's $\frac 1 4$th theorem does not generalize to the operator setting \cite{HP09}.

In \cite{PW14} we give with $D(g)=D(\Lambda)\geq B(S)$ a necessary condition on the (operator) sampling rate based  on the bandwidth $B(S)$ of $OPW(S)$   which is defined as
\begin{align}\label{eqn:bandwidth}
 B(S) = \sup_{t\in\R}|\supp\eta(t,\nu)|  = \Big\|\int_\R \chi_{S}(\cdot,\nu)\,d\nu\Big\|_\infty.
\end{align}
Here,  $\chi_S$ denotes  the characteristic function of $S$.  This quantity can be interpreted as the maximum vertical extent of $S$ and takes into account gaps in $S$. Moreover, in $\cite{PW14}$ we discuss the goal of constructing $\{\lambda_n\}$ of small density, and/or large gaps in order to reserve time-slots for information transmission.  Results in this direction can be interpreted as giving bounds on the capacity of a time-variant channel in $OPW(S)$ in terms of $|S|$
\cite{PW14}.

Finally, we give in \cite{PW14} an example of an operator class $OPW(S)$ that cannot be identified by any identifier of the form $g=\sum_{n\in\Z} c_n \delta_{nT}$
with $T>0$ and periodic $c_n$, but requires coefficients that form a bounded but non-periodic sequence.  In this case, $S$ is a parallelogram and
$B(S)=D(g)$ (see Figure~\ref{fig:nonperiodic})

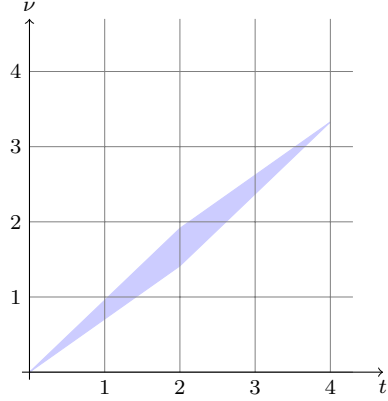
\begin{figure}
  \centering
\begin{tikzpicture}[scale=0.5]
\def\L{9}
\def\Lminusone{8}
\filldraw [fill=blue!20, draw=blue!20] (0,0) -- (4,2.828) -- (8,6.657)-- (4,3.828) -- (0,0);
\draw[style=help lines, step={(2,2)}] (0,0) grid (\L-0.4, \L+0.4);
\draw[->] (-0.2,0) -- (\L+0.4,0) node[below] {$t$};
\draw[->] (0,-0.2) -- (0,\L+0.4) node[above] {$\nu$};
\node[below] at (2,0) {$1$};
\node[below] at (4,0) {$2$};
\node[below] at (6,0) {$3$};
\node[below] at (8,0) {$4$};
\node[left] at (0,2) {$1$};
\node[left] at (0,4) {$2$};
\node[left] at (0,6) {$3$};
\node[left] at (0,8) {$4$};
\end{tikzpicture}


  \caption{ The the operator class $OPW^2(S)$ with $S=  ( 2,\  2 \ ;\  \sqrt{2}, \ \sqrt{2} +1/2)[0,1]^2$ whose
  area equals $1$ and bandwidth equals $1/2$ is identifiable by a (non-periodically) weighted delta train with sampling density $1/2$.
  It is not identifiable using a periodically-weighted delta train.
  }
\label{fig:nonperiodic}
\end{figure}

\subsection{Sampling of $OPW(S)$ with unknown $S$.}

In some applications, it is justified to assume that the set $S$ has small area, but its shape and location are unknown.  If further $S$ satisfies some  basic geometric assumptions that guarantee that $S$ is contained in $[0,TP]\times[-1/2T, 1 / 2T]$ and only meets few rectangles of the rectification grid $[kT,(k+1)T]\times [q /TP,(q+1)/TP]$, then recovery of $S$ and, hence, an operator in $OPW(S)$ is  possible \cite{PW14,HB13}.

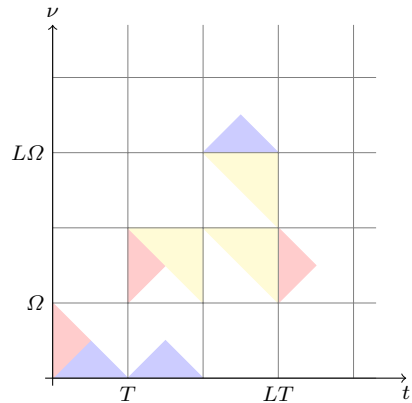
\begin{figure} \label{figure:colors}
  \centering
\begin{tikzpicture}[scale=0.5]
\def\L{9}
\def\Lminusone{8}
\filldraw [fill=blue!20, draw=blue!20] (0,0) -- (1,1) -- (2,0) -- (0,0);
\filldraw [fill=blue!20, draw=blue!20] (2,0) -- (3,1) -- (4,0) -- (2,0);
\filldraw [fill=blue!20, draw=blue!20] (4,6) -- (5,7) -- (6,6) -- (4,6);

\filldraw [fill=red!20, draw=red!20] (0,0) -- (1,1) -- (0,2) -- (0,0);
\filldraw [fill=red!20, draw=red!20] (2,2) -- (3,3) -- (2,4) -- (2,2);
\filldraw [fill=red!20, draw=red!20] (6,2) -- (7,3) -- (6,4) -- (6,2);

\filldraw [fill=yellow!20, draw=yellow!20] (6,4) -- (4,6) -- (6,6) -- (6,4);
\filldraw [fill=yellow!20, draw=yellow!20] (4,2) -- (2,4) -- (4,4) -- (4,2);
\filldraw [fill=yellow!20, draw=yellow!20] (6,2) -- (4,4) -- (6,4) -- (6,2);

\draw[style=help lines, step={(2,2)}] (0,0) grid (\L-0.4, \L+0.4);
\draw[->] (-0.2,0) -- (\L+0.4,0) node[below] {$t$};
\draw[->] (0,-0.2) -- (0,\L+0.4) node[above] {$\nu$};
\node[below] at (2,0) {$T$};
\node[left] at (0,2) {$\Omega$};
\node[below] at (6,0) {$LT$};
\node[left] at (0,6) {$L\Omega$};
\end{tikzpicture}
  \caption{For $S$ the union of the colored sets, $OPW(S)$ is identifiable even though $7 > 3$ boxes are active, implying that
  $S$ cannot be rectified with $P=3$ and $T=1$ is not possible (see Section~\ref{section:mainidentification}).
  Recovering $\eta$ from $Hg$ requires solving three systems of linear equations, one to recover $\eta$ on the yellow support set, one to recover $\eta$ on the red support set, and one to recover $\eta$ on the blue support set. The reconstruction formula (\ref{eqn:reconstructionformula}) does not apply for this set $S$.  }
\label{fig:rectification}
\end{figure}

The independently obtained results in \cite{PW14,HB13} employ the same identifiers $g=\sum_{n\in\Z} c_n \delta_{\lambda_n}$ as introduced above. Operator identification is therefore again reduced  to solving  \eqref{eqn:basiclinearsystem}, that is, the system of $P$ linear equations
\begin{equation}\label{eqn:basiclinearsystemsimple}
{\bf Z}(t,\nu) = G(c)\boldsymbol\eta (t,\nu)
\end{equation}
for the vector $\boldsymbol \eta (t,\nu)\in\mathbb C^{P^2}$ for $(t,\nu)\in [0,T]\times[-1/2TP, 1/2TP]$. While the zero components of
$\boldsymbol \eta (t,\nu)$ are not known, the vector is known to be very sparse.  Hence, for fixed $(t,\nu)$, we can use the fact that $G(c)$ is full spark and recover $\boldsymbol \eta (t,\nu)$ if it has at most $P/2$ nonzero entries. Indeed, assume  ${\boldsymbol{\eta}}(t,\nu)$ and $\widetilde{\boldsymbol{\eta}}(t,\nu)$ solve \eqref{eqn:basiclinearsystemsimple}
and  both have at most $P/2$ nonzero entries. Then $\boldsymbol {\eta} (t,\nu)-\widetilde{\boldsymbol{\eta}} (t,\nu)$ has at most $P$ nonzero entries and the fact that $G(c)$ is full spark indicates that  $G(c)( \boldsymbol {\eta} (t,\nu)-\widetilde{\boldsymbol{\eta}} (t,\nu))=0$  implies  $\boldsymbol {\eta} (t,\nu)-\widetilde{\boldsymbol{\eta}} (t,\nu)=0$.

Clearly, under the geometric assumptions alluded to above, the criterion that at most $P/2$ rectangles in the grid are met can be translated to the unknown area of $S$ has measure less than or equal to 1/2.

In \cite{HB13}, this area 1/2 criterion is improved by showing that $H$ can be identified whenever at most $P-1$ rectangles in the rectification grid are met by
$S$. This result is achieved by using a joint sparsity argument, based on the assumption  that for all $(t,\nu)$, the same cells are active.

Alternatively, the ``area 1/2'' result can be strengthened by not assuming that for all $(t,\nu)$, the same cells are active. This requires solving \eqref{eqn:basiclinearsystemsimple}, for  $\boldsymbol{\eta} (t,\nu)$ sparse, for each considered $(t,\nu)$ independently, see Figure~\ref{fig:rectification} and \cite{PW14}.

It must be added though, that solving  \eqref{eqn:basiclinearsystemsimple} for $\boldsymbol{\eta} (t,\nu)$ being $P/2$ sparse is not possible for moderately sized $P$, for example for $P > 15$. If we further reduce the number of active boxes, then compressive sensing algorithms such as Basis Pursuit and Orthogonal Matching Pursuit become available, as is discussed in the following section.

\subsection{Finite dimensional operator identification and compressive sensing}

Operator sampling in in the finite dimensional setting translates into the following matrix probing problem \cite{PRT08,CD12,BD14}. For a class of matrices $\boldsymbol{\mathcal M}\in\mathbb C^{P\times P}$, find $c\in \mathbb C^P$ so that we can recover $M\in\boldsymbol{\mathcal M}$ from $Mc$.

\begin{figure}
\begin{center}
 \begin{minipage}[b][2cm][t]{3cm}
  \includegraphics[width=3cm]{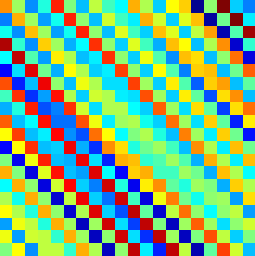}
\end{minipage} \
 \begin{minipage}[b][2cm][t]{0.2cm}
 \includegraphics[height=3cm]{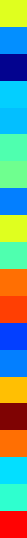}
 \end{minipage}
\begin{minipage}[c][2cm][t]{0.6cm}
\begin{center}

   \
   \\[.3cm]

   ${=}$
\end{center}
\end{minipage}
  \begin{minipage}[b][2cm][t]{.2cm}\includegraphics[height=3cm]{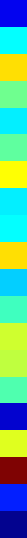}
  \end{minipage}
\end{center}
 \caption{The matrix probing problem: find $c$ so that the map $\boldsymbol{\mathcal M} \longrightarrow \mathbb C^P$, $M\mapsto Mc$ is injective and therefore invertible.}
\end{figure}

 The classes of operator considered here are of the form $M_{\boldsymbol \eta}=\sum_\lambda \boldsymbol{\eta}_\lambda B_\lambda$ with $B_{\lambda}=B_{p,q}=\mathcal T^p \mathcal M^q$, and the matrix identification problem is reduced to  solving
 \begin{equation}\label{eqn:basiclinearsystemsimpleno}\boldsymbol Z=M_{\boldsymbol{\eta}} c=\sum_{p,q=0}^{P-1} {\boldsymbol{\eta}}_{p,q} \big(\ \mathcal T^p \mathcal M^q c\big) =G(c)\boldsymbol{\eta},
 \end{equation}
where $c$ is chosen appropriately; this is just   \eqref{eqn:basiclinearsystemsimple} with the dependence on $(t,\nu)$  removed.

If $\boldsymbol\eta$ is assumed to be $k$-sparse, we arrive at the classical compressive sensing problem with measurement matrix 
$G(c)\in \mathbb C^{P \times P^2}$ which depends on $c=(c_0,c_1,\ldots,c_{P-1})$. To achieve recovery guarantees for Basis Pursuit and Orthogonal Matching Pursuit, averaging arguments have to be used that yield results on the expected qualities of $G(c)$.  This problem was discussed in \cite{PRT08,PR09,PRT13} as well as, in slightly different terms, in \cite{AHS08,HS09}. The strongest results were achieved in \cite{KMR14} by estimating Restricted Isometry Constants  for $c$ being a Steinhaus sequence. These results show that with high probability, $G(c)$ has the property that Basis Pursuit recovers $\boldsymbol{\eta}$ from   $G(c)\boldsymbol{\eta}$ for every $k$ sparse $\boldsymbol{\eta}$ as long as  $k \leq C \, P / \log^2 P$. for some universal constant $C$.

\begin{figure}
 \includegraphics[width=11.75cm]{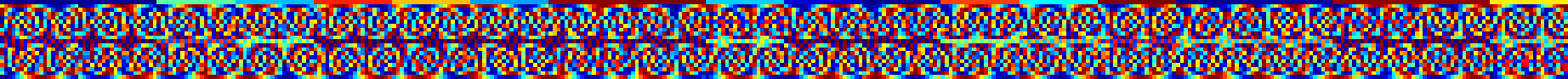}
 \caption{Time-frequency structured measurement matrix $G(c)$ with $c$ randomly chosen.}
\end{figure}

\subsection{Stochastic operators and channel estimation}

It is common that models of wireless channels and radar environments  take the stochastic nature of the medium into account. In such models, the spreading function $\eta(t,\nu)$ (and therefore the operator's kernel and Kohn--Nirenberg symbol) are random processes, and the operator is split into the sum of its deterministic portion, representing the mean behavior of the channel, and its zero-mean stochastic portion that represents the noise and the environment.

\begin{figure}
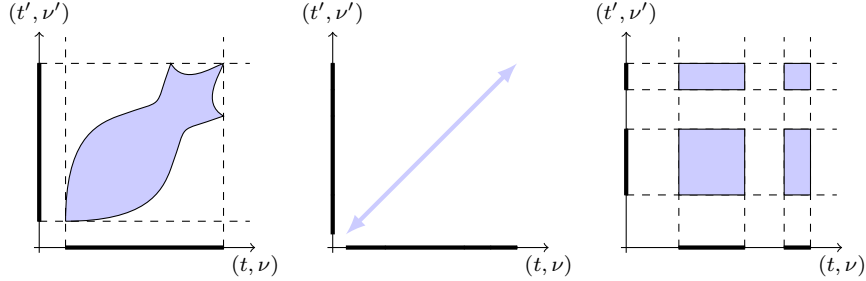

\tikz[scale=0.35]  \curvysupport;
\tikz[scale=0.35]  \wssussupport;
\tikz[scale=0.35]  \tensorsupport;
 \caption{Support sets of autocorrelation functions, the general case, the WSSUS case, and the tensor case.}
\end{figure}

The detailed analysis of the stochastic case was carried out in \cite{PZ14a,PZ14b}. One of the foci of these works lies in the goal of determining  the second-order statistics of the (zero mean) stochastic process $\eta(\tau,\nu)$, that is, its so called covariance function $R(\tau,\nu,\tau',\nu') = \mathbb E \{\eta(\tau,\nu)\, \overline{\eta(\tau',\nu')}\}$.
In   \cite{PZ14a,PZ14b}, it was shown that a necessary but not sufficient condition for the identifiability of $R\eta(\tau,\nu,\tau',\nu')$ from the output covariance
$
A(t,t') = \mathbb E \{ H g(t) \, \overline{H g(t')}\}
$
is that $R(\tau,\nu,\tau',\nu')$ is supported on a bounded set of 4-dimensional volume less than or equal to
one. Unfortunately, for some sets $S\subseteq \R^4$ of arbitrary small measure,  the respective stochastic
operator Paley--Wiener space $StOPW(S)$ of operators with $R\eta$ supported on $S$ is not identifiable; this is a striking
difference to the deterministic setup where the geometry of $S$ does not play a role at all.

 In \cite{OPZ,PZ14c}  the special case of   \emph{wide-sense stationary operators with uncorrelated scattering}, or WSSUS operators is considered.  These operators are characterized by the property that
\begin{equation*}\label{eq:scattering}
R\eta(t,\nu,t',\nu') = C_\eta (t,\nu) \, \delta(t-t') \, \delta(\nu-\nu').
\end{equation*}
The function $C_\eta(t,\nu)$ is then called \emph{scattering function} of $H$.
Our results on the identifiability of stochastic operator classes allowed for the construction of two estimators for  scattering functions \cite{OPZ,PZ14c}.  The estimator given in \cite{PZ14c} is applicable, whenever the scattering function of $H$ has bounded support.  Note that the autocorrelation of a WSSUS operator is supported on a two dimensional plane in $\R^4$ which therefore has 4D volume 0, a fact that allows us to lift commonly assumed restrictions on the size of the 2D area of the support of the scattering function.

For details, formal definitions of identifiability and detailed statements of results we refer to the papers \cite{OPZ,PZ14a,PZ14b,PZ14c}.

\section{Appendix:  Proofs of Theorems.}

\subsection{Proof of Lemma~\ref{lem:matrixequationquasiperiodic}}\label{section:proof of lemma}

In order to see how the time-frequency shifts of $c$ arise, we will briefly outline the calculation that leads to (\ref{eqn:matrixequationquasiperiodic}).
It can be seen by direct calculation using the representation given by (\ref{eqn:operator3}), that  if $g = \sum_{n}\delta_{nTP}$ then
$\ip{Hg}{s} = \ip{\eta_H}{Z_{TP} s}$ for all $s\in{\cal S}(\R)$ where the bracket on the left is the $L^2$ inner product on
$\R$ and that on the right the $L^2$ inner product on the rectangle $[0,TP]{\times}[0,1/(TP)]$.  Periodizing the integral on the left gives
\begin{align*}
\ip{\eta_H}{Z_{TP} s} = & \int_0^{1/(TP)}\int_0^{TP} \sum_{k}\sum_{m} \eta_H(t+kTP,\nu+m/(TP)) \\
&  \hskip.5in  e^{-2\pi i\nu kTP}\overline{Z_{TP} s(t,\nu)}\,dt\,d\nu.
\end{align*}
Since this holds for every $s\in{\cal S}(\R)$, we conclude that
\begin{align*}
(Z_{TP}\circ H)g(t,\nu) & \\
  & \hskip-.5in = 1/(TP)\,\sum_{k}\sum_{m} \eta_H(t+kTP,\nu+m/(TP))\,e^{-2\pi i\nu kTP}.
\end{align*}

Given $g = \sum_{n\in\Z} c_n\,\delta_{nT}$, for a period-$P$ sequence $c=(c_n)$,
and letting $n=mP-q$ for $m\in\Z$ and $0\le q<P$, we obtain
\begin{align*}
g  &  =    \sum c_n\,\delta_{nT} =    \sum_{q=0}^{P-1} \sum_{m\in\Z} c_{mP-q}\,\delta_{mPT - qT} \\
   &  =    \sum_{q=0}^{P-1}  c_{-q} \mathcal T_{-qT}\,\bigg(\sum_{m\in\Z}\,\delta_{mPT}\bigg).
\end{align*}
Since for $\alpha\in\R$, the spreading function of $H\circ \mathcal T_\alpha$ is $\eta_H(t-\alpha,\nu)\,e^{2\pi i\nu\alpha}$,
we arrive at
\begin{align}
(Z_{TP}\circ H)g(t,\nu) & \nonumber \\
    & \hskip-.75in = 1/(TP)\,\sum_{q=0}^{P-1} c_{-q}\,\sum_{k}\sum_{m} \eta_H(t+kTP+qT,\nu+m/(TP)) \nonumber \\
    & \hskip.75in e^{-2\pi i(\nu+m/(TP))qT}\,e^{-2\pi i\nu kTP}. \label{eqn:weighteddeltatrain}
\end{align}

Letting $m=jP+\ell$ for $j\in\Z$ and $0\le \ell<P$, we obtain
\begin{align*}
(Z_{TP}\circ H)g(t,\nu) & \nonumber \\
    & \hskip-.75in = 1/(TP)\,\sum_{q=0}^{P-1} c_{-q}\,\sum_{k}\sum_{j}\sum_{\ell=0}^{P-1} \eta_H(t+kTP+qT,\nu+j/T+\ell/(TP)) \nonumber \\
    & \hskip.75in e^{-2\pi i\nu qT}\,e^{-2\pi i\ell q/P}\,e^{-2\pi i\nu kTP} \nonumber \\
    & \hskip-.75in = 1/(TP)\,\sum_{q=0}^{P-1}\sum_{\ell=0}^{P-1} \big(c_{-q}\,e^{-2\pi i\ell q/P}\big)\,
        e^{-2\pi i\nu qT}\,\eta^{Q}_H(t + T q,\nu + \ell/TP).
\end{align*}

Finally, replacing $t$ by $t+pT$ for $p=0,\,1,\,\dots,\,P{-}1$, and changing indices by replacing $q$ by $q-p$, we obtain
\begin{align*}
(Z_{TP}\circ H)g(t+pT,\nu) &  \\
    & \hskip-1.25in = 1/(TP)\,\sum_{q=0}^{P-1}\sum_{\ell=0}^{P-1} \big(c_{-q}\,e^{-2\pi i\ell q/P}\big)
         \,e^{-2\pi i\nu qT}\,\eta^{Q}_H(t + (q+p) T,\nu + \ell/TP)  \\
    & \hskip-1.25in = 1/(TP)\,\sum_{q=0}^{P-1}\sum_{\ell=0}^{P-1} \big(c_{-(q-p)}\,e^{-2\pi i\ell (q-p)/P}\big) \nonumber \\
    & \hskip.75in e^{-2\pi i\nu (q-p)T}\,\eta^{Q}_H(t + q T,\nu + \ell/TP).
\end{align*}
The observation that $(\mathcal T^q\,\mathcal M^m c)_p = c_{p-q}\,e^{2\pi i m (p-q)/P}$ completes the proof.

\subsection{Proof of Theorem~\ref{thm:genericfullsparkprime}}\label{section:prime}

To see why this is true, define $\mu(M)$ to be the number of
diagonals of $M$ whose product is a multiple of $p_M$, and proceed by induction on the size of the matrix $M$.
If $M$ is $1\times 1$ then the result is obvious.  Suppose that $M$ is $n\times n$ and that it is described by the vector
$\ell=(\ell_0,\,\dots,\,\ell_{P-1})$.
Assuming without loss of generality that the variable of smallest index in $p_M$ with a nonzero
exponent is $c_0$, there is a row of $M$ in which
the variable $c_0$ appears $ \ell_j$ times for some index $j$.
Choose one of these terms and delete the row and column in which
it appears. Call the remaining matrix $M'$.  The vector $\ell$ describing $M'$ is
$(\ell_0,\,\dots,\,\ell_{j-1},\,\ell_j-1,\,\ell_{j+1},\,\dots,\,\ell_{P-1})$,
and is independent of which term was chosen from the given row to
form $M'$.   By the construction of the LI monomial,
$p_M=c_0\,p_{M'}$ and by the induction hypothesis
$$\mu(M') = \ell_0!\,\cdots\,\ell_{j-1}!\,(\ell_j-1)!\,\ell_{j+1}!\,\cdots\,\ell_{P-1}!.$$
Since there are $\ell_j$ ways to choose a term from the given row
to produce $M'$ we have that
$$\mu(M) = \ell_j\,\mu(M') = \ell_0!\,\cdots\,\ell_{j-1}!\,\ell_j( l_j-1)!\,\ell_{j+1}!\,\cdots\,\ell_{P-1}! = \prod_{\kappa=0}^{P-1} \ell_\kappa!$$
which was to be proved.

Since each term $a_\sigma\,C^\sigma$ in (\ref{eqn:determinant}) is made up of a sum of precisely this many terms, it follows that
exactly one of these terms is a multiple of the LI monomial.
Alternatively, we can think of the LI monomial as the one corresponding to the $\sigma\in S_P/\Gamma$ that minimizes the functional
$\Lambda_0(C^\sigma) = \sum_{i=0}^{L-1} i^2\,H(\alpha_i)$ where $\alpha_i$ is the exponent of $c_i$ in $C^\sigma$ and where $H(\alpha_i)=0$
if $\alpha_i=0$ and $1$ otherwise.

Because by Chebotarev's Theorem, $a_\sigma\ne 0$ for all $\sigma$ the proof works for any square submatrix $M$, no matter what size.
This gives us Theorem~\ref{thm:genericfullsparkprime}.

\subsection{Proof of Theorem \ref{thm:genericfullspark}}\label{sec:malikiosis}
We first need to assert the existence of a cyclical renumbering of the variables such that with respect to the new
trivial partition $A'=(A_\kappa')_{\kappa=0}^{P-1}$, the CI monomial is given by
$$C^I = \prod_{\kappa=0}^{P-1}\,\prod_{j\in A_\kappa'} c_{j-\kappa}$$
in other words, if $j\in A_\kappa'$ then $0\le j-\kappa <P$.  Note first that since $\min(A_\kappa')=\sum_{i=0}^{\kappa-1}\ell_i'$
for all $\kappa$, $j\in A_\kappa'$ implies that $j\ge\sum_{i=0}^{\kappa-1}\ell_i'$.  Therefore, it will suffice to find a $0\le\gamma<P$
such that for all $\kappa$, $\sum_{i=0}^{\kappa-1}\ell_i' - \kappa\ge 0$ so that $j-\kappa \ge \sum_{i=0}^{\kappa-1}\ell_i' - \kappa\ge 0$.

Let $0\le \gamma<P$ be such that the quantity $\sum_{i=0}^{\gamma-1} \ell_i - \gamma$ is minimized, let
$$\ell' = (\ell_i')_{i=0}^{L-1} = (\ell_{(i+\gamma)mod\ P})_{i=0}^{P-1},$$
and let $A' = (A_\kappa')_{\kappa=0}^{P-1}$ be the corresponding trivial partition.  Now fix $\kappa$ and assume that $\kappa+\gamma \le P$.
Then
\begin{eqnarray*}
\sum_{i=0}^{\kappa-1}\ell_i' - \kappa
&  =  &  \sum_{i=0}^{\kappa-1}\ell_{(i+\gamma)} - \kappa \\
&  =  &  \bigg(\sum_{i=0}^{\kappa+\gamma-1}\ell_i - (\kappa+\gamma)\bigg) - \bigg(\sum_{i=0}^{\gamma-1}\ell_i - \gamma\bigg) \\
& \ge & 0
\end{eqnarray*}
since the second term in the difference is minimal.  If $\kappa+\gamma \ge P+1$ then remembering that $\sum_{i=0}^{P-1}\ell_i = L$
\begin{eqnarray*}
\sum_{i=0}^{\kappa-1}\ell_i' - \kappa
&  =  &  \sum_{i=0}^{\kappa-1}\ell_{(i+\gamma)mod\ P} - \kappa \\
&  =  &  \sum_{i=\gamma}^{P-1} \ell_i + \sum_{i=0}^{\kappa+\gamma-P-1}\ell_i - \kappa \\
&  =  &  \sum_{i=0}^{P-1} \ell_i - \sum_{i=0}^{\gamma-1} \ell_i + \sum_{i=0}^{\kappa+\gamma-P-1}\ell_i - \kappa \\
&  =  &  \bigg(\sum_{i=0}^{(\kappa+\gamma-P)-1}\ell_i - (\kappa+\gamma-P)\bigg)  - \bigg(\sum_{i=0}^{\gamma-1}\ell_i - \gamma\bigg) \\
& \ge & 0.
\end{eqnarray*}

In order to complete the proof, we must show that $\Lambda(C^\sigma)\ge\Lambda(C^I)$ for all $\sigma\in S_P/\Gamma$ with equality
holding if and only if $\sigma$ is trivial.  This will follow by direct calculation together with the following lemma which follows from a classical result on
rearrangements of series (\cite{HLP52}, Theorems~368, 369).  This result is Lemma~3.3 in \cite{M13}.

First, however, we adopt the following notation.  For $0\le n<P$, let $b_n = \kappa$ if $n\in A_\kappa$.  With this notation, given $\sigma\in S_P/\Gamma$,
$$C^\sigma = \prod_{n=0}^{P-1} c_{(\sigma(n)-b_n)\ mod\ P}$$
and under the above assumptions,
$$C^I = \prod_{n=0}^{P-1} c_{(n-b_n)}.$$
Moreover,
\begin{eqnarray*}
\Lambda(C^\sigma)
&  =  &  \sum_{i=0}^{P-1} i^2\,\alpha_i \\
&  =  &  \sum_{i=0}^{P-1} i^2\,(\#\{n\colon (\sigma(n)-b_n)\ mod\ P = i\}) \\
&  =  &  \sum_{i=0}^{P-1} \big( (\sigma(n)-b_n)\ mod\ P \big)^2.
\end{eqnarray*}

\begin{lemma}\label{lem:rearrangement}
Given two finite sequences of real numbers $(\alpha_n)$ and $(\beta_n)$ defined up to rearrangement,
the sum
$$\sum_n \alpha_n\,\beta_n$$
is maximized when $\alpha$ and $\beta$ are both monotonically increasing or monotonically decreasing.
Moreover, if for every rearrangement $\alpha'$ of $\alpha$,
$$\sum_n \alpha_n'\,\beta_n \le \sum_n \alpha_n\beta_n$$
then $\alpha$ and $\beta$ are {\em similarly ordered}, that is, for every $j,\,k$,
$$(\alpha_j-\alpha_k)(\beta_j-\beta_k)\ge 0.$$

In particular, for every $\sigma\in S_P$,
$$\sum_{n=0}^{P-1} n\,b_n \ge \sum_{n=0}^{P-1} \sigma(n)\,b_n$$
with equality holding if and only if $\sigma$ is trivial.
\end{lemma}

\begin{proof}
The first part of the lemma is simply a restatement of Theorems~368 and 369 of \cite{HLP52}.  To prove the second part,
note first that $b_n$ is a non-decreasing sequence and in particular is constant on each $A_\kappa$.
Theorem~368 in \cite{HLP52} states that a sum of the form $\sum_{n=0}^{P-1} \sigma(n)\,b_n$ is maximized when $\sigma(n)$ is monotonically increasing,
which proves the given inequality.  Since $b_n$ is constant on each $A_\kappa$, it follows that if $\sigma$ is trivial, then
we have equality.

If $\sigma$ is not trivial then we will show that the sequences $\sigma(n)$ and $b_n$ are not similarly ordered.  Letting $\kappa$ be the minimal index
such that $A_\kappa$ is not left invariant by $\sigma$, there exists $m\in A_\kappa$ such that $\sigma(m)\in A_\mu$ for some $\mu>\kappa$, and for some
$\lambda>\kappa$ there exists $k\in A_\lambda$ such that $\sigma(k)\in A_\kappa$.  Therefore, $b_m=\kappa < \lambda = b_k$ but since $\mu>\kappa$,
$\sigma(m) > \sigma(k)$, and so $\sigma(n)$ and $b_n$ are not similarly ordered.
\end{proof}

In order to complete the proof, define $\C_1,\,\C_2\subseteq\{0,\,\dots,\,P-1\}$ by $n\in\C_1$ if $0\le \sigma(n)-b_n <P$,
and $n\in\C_2$ if $-P+1\ge \sigma(n)-b_n <0$ (note that always $\abs{\sigma(n)-b_n} <P$) so that when $n\in C_2$,
$(\sigma(n)-b_n)\ mod\ P = \sigma(n)-b_n + P$.  Let $\sigma'(n)=\sigma(n)$ if $n\in\C_1$ and $\sigma(n)+P$ if $n\in\C_2$, and let
$(a_n)_{n=0}^{P-1}$ be an increasing sequence enumerating the set $\sigma(\C_1)\cup(\sigma(\C_2)+P)$.  Therefore,
\begin{align*}
\Lambda(C^\sigma) - \Lambda(C^I)
&  =    \sum_{n=0}^{P-1} (\sigma'(n)-b_n)^2 - \sum_{n=0}^{P-1} (n-b_n)^2 \\
&  =    \bigg[\sum_{n=0}^{P-1} (\sigma'(n)-b_n)^2 - \sum_{n=0}^{P-1} (a_n-b_n)^2\bigg]
\\ & \qquad \qquad+ \bigg[\sum_{n=0}^{P-1} (a_n-b_n)^2 - \sum_{n=0}^{P-1} (n-b_n)^2\bigg] \\
&  =    2\,\bigg[\sum_{n=0}^{P-1} a_nb_n - \sigma'(n)b_n\bigg] + \bigg[\sum_{n=0}^{P-1} (a_n-b_n)^2 - (n-b_n)^2\bigg] \\
&  =    I + II.
\end{align*}
Since $a_n$ is increasing, $I\ge 0$ by Lemma~\ref{lem:rearrangement}, and since $a_n\ge n$ for all $n$,
$(a_n-b_n)\ge(n-b_n)\ge 0$ so that $(a_n-b_n)^2\ge (n-b_n)^2$ and hence $II\ge 0$.  It remains to show that equality holds only if $\sigma$ is trivial.
If $\Lambda(C^\sigma)=\Lambda(C^I)$ then $I = II = 0$.  Since $II=0$, $\C_2=\emptyset$ for if $a_n\in\sigma(\C_2)+P$ then $a_n>n$ and we would have $II > 0$.
Since $\C_2=\emptyset$, $\sigma'(n)=\sigma(n)$ so that
\begin{eqnarray*}
0 &  =  &  \Lambda(C^\sigma) - \Lambda(C^I) \\
  &  =  &  \sum_{n=0}^{P-1} (\sigma(n)-b_n)^2 - \sum_{n=0}^{P-1} (n-b_n)^2 \\
  &  =  &  2\,\sum_{n=0}^{P-1} (n\,b_n - \sigma(n)\,b_n)
\end{eqnarray*}
which by Lemma~\ref{lem:rearrangement} implies that $\sigma$ is trivial.  The proof is complete.

\end{document}